\documentclass[twoside,onecolumn]{article}

\usepackage[english]{babel} 

\usepackage[hmarginratio=1:1,top=32mm,columnsep=20pt]{geometry} 
\usepackage[hang, small,labelfont=bf,up,textfont=it,up]{caption} 
\usepackage{booktabs} 

\usepackage{geometry} 
\usepackage{graphicx} 

\usepackage{booktabs} 
\usepackage{array} 
\usepackage{paralist} 
\usepackage{verbatim} 
\usepackage{subfig} 
\usepackage{amsmath}	
\allowdisplaybreaks[1]	
\numberwithin{equation}{section}	
\usepackage{amsfonts}	
\usepackage{amsthm}	
\usepackage{cases}
\usepackage{bbm}		
\usepackage{braket}	
\usepackage{mathtools}	
\usepackage[english]{varioref}	
\usepackage[multiple]{footmisc}	

\usepackage{fancyhdr} 
\theoremstyle{plain}
\newtheorem{example}{Example}[section]
\theoremstyle{plain}
\newtheorem{remark}{Remark}[section]
\theoremstyle{plain}
\newtheorem{prop}{Proposition}[section]
\theoremstyle{plain}
\newtheorem{lemma}{Lemma}[section]
\theoremstyle{plain}
\newtheorem{theorem}{Theorem}[section]
\theoremstyle{plain}
\newtheorem{corollary}{Corollary}[section]
\theoremstyle{remark}
\newtheorem{definition}{Definition}[section]
\theoremstyle{assumption}
\newtheorem{assumption}{Assumption}[section]

\DeclarePairedDelimiter{\abs}{\lvert}{\rvert}

\DeclareMathOperator*{\var}{\text{var}}
\DeclareMathOperator*{\esssup}{ess\,sup}

\usepackage{sectsty}
\allsectionsfont{\sffamily\mdseries\upshape} 

\usepackage{abstract} 

\usepackage{titlesec} 
\titleformat{\section}[block]{\large\bfseries\centering}{\thesection.}{1em}{} 
\titleformat{\subsection}[block]{\large\bfseries}{\thesubsection.}{1em}{} 

\usepackage[nottoc,notlof,notlot]{tocbibind} 
\usepackage[titles,subfigure]{tocloft} 

\usepackage[a-1b]{pdfx} 
\usepackage{titling} 
\usepackage{hyperref}

\hypersetup{%
    pdfpagemode={UseOutlines},
    bookmarksopen,
    pdfstartview={FitH},
    colorlinks,
    linkcolor={black},
    citecolor={black},
    urlcolor={black}
  }

\pretitle{\begin{center}\Huge\bfseries} 
\posttitle{\end{center}} 
\title{Optimal proportional reinsurance and investment for stochastic factor models} 
\author{%
\textsc{Brachetta M.}\thanks{Department of Economics, University of Chieti-Pescara, Viale Pindaro, 42 - 65127 Pescara, Italy.} \\[1ex]
\normalsize \href{mailto:matteo.brachetta@unich.it}{matteo.brachetta@unich.it}
\and
\textsc{Ceci, C.}\footnotemark[1] \thanks{Corresponding author.}\\[1ex]
\normalsize \href{mailto:c.ceci@unich.it}{c.ceci@unich.it}
}
\date{} 

\providecommand{\keywords}[1]{\textbf{\textit{Keywords:}} #1}
\providecommand{\jelcodes}[1]{\textbf{\textit{JEL Classification codes:}} #1}
\providecommand{\msccodes}[1]{\textbf{\textit{MSC Classification codes:}} #1}

\begin{document}
\maketitle

\begin{abstract}
\noindent In this work we investigate the optimal proportional reinsurance-investment strategy of an insurance company which wishes to maximize the expected exponential utility of its terminal wealth in a finite time horizon. Our goal is to extend the classical Cram\'er-Lundberg model introducing a stochastic factor which affects the intensity of the claims arrival process, described by a Cox process, as well as the insurance and reinsurance premia. Using the classical stochastic control approach based on the Hamilton-Jacobi-Bellman equation we characterize the optimal strategy and provide a verification result for the value function via classical solutions of two backward partial differential equations. Existence and uniqueness of these solutions are discussed.
Results under various premium calculation principles are illustrated and a new premium calculation rule is proposed in order to get more realistic strategies and to better fit our stochastic factor model. Finally, numerical simulations are performed to obtain sensitivity analyses.
\end{abstract}

\noindent\keywords{Optimal proportional reinsurance, optimal investment, Cox model, stochastic control.}\\
\noindent\jelcodes{G220, C610, G110.}\\
\noindent\msccodes{93E20, 91B30, 60G57, 60J75.}\\
\noindent \textit{Declarations of interest: none.}

\section{Introduction}

In this paper we investigate the optimal reinsurance-investment problem of an insurance company which wishes to maximize the expected exponential utility of its terminal wealth in a finite time horizon.
In the actuarial literature there is an increasing interest in both optimal reinsurance and optimal investment strategies, because they allow insurance firms to increase financial results and to manage risks. In particular, reinsurance contracts help the reinsured to increase the business capacity, to stabilize operating results, to enter in new markets, and so on. Among the traditional reinsurance arrangements the excess-of-loss and the proportional treaties are of great importance. The former was studied in~\cite{shengrongzhao:optreins},~\cite{lietal:robustXL} and references therein. The latter was intensively studied by many authors under the criterion of maximizing the expected utility of the terminal wealth. Beyond the references contained therein, let us recall some noteworthy papers: in~\cite{liuma:optreins} the authors considered a very general model, also including consumption, focusing on well posedness of the optimization problem and on existence of admissible strategies; in~\cite{liangetal:optreins} a stock price with instantaneous rate of investment return described by an Ornstein-Uhlenbeck process has been considered ; in~\cite{liangbayraktar:optreins} the problem has been studied in a partially observable framework by introducing an unobservable Markov-modulated risk process; in~\cite{zhuetal:reins_defaultable} the surplus is invested in a defaultable financial market; in~\cite{liangetal:commonshock} and~\cite{yuen_liang_zhou:commonshock} multiple dependent classes of insurance business are considered. All these works may be considered as attempts to extend both the insurance risk and the financial market models.
In all these articles we can recognize two different approaches to dealing with the surplus process of the insurance company: some authors considered it as a diffusion process approximating the pure-jump term of the Cram\'er-Lundberg model (see for example~\cite{baiguo:optreins,caowan:optreins,zhangyu:optreins,gu_yang:CEVoptreins,lietal:robustXL} and references therein).
This approach is validated by means of the famous Cram\'er-Lundberg approximation (see~\cite{grandell:risk}). Other authors (see~\cite{liuma:optreins,zhuetal:reins_defaultable,liangetal:optreins,shengrongzhao:optreins,yuen_liang_zhou:commonshock} and references therein) took into account the jump term using a compound Poisson risk model with constant intensity, that is the classical Cram\'er-Lundberg model. On the one hand this is the standard model for nonlife insurance and it is simple enough to perform calculations, on the other it is too simple to be realistic (as noticed by~\cite{hipp:stoch.control_applications}).\\
As observed by Grandell, J. in~\cite{grandell:risk}, more reasonable risk models should allow the insurance firm to consider the so called \textit{size fluctuations} as well as the \textit{risk fluctuations}, which refer respectively to variations of the number of policyholders and to modifications of the underlying risks.

This paper aims at extending the classical risk model by modelling the claims arrival process as a doubly stochastic Poisson process with intensity affected by an exogenous stochastic process $\{Y_t\}_{t\in[0,T]}$. This environmental factor lead us to a reasonably realistic description of any risk movement (see~\cite{grandell:risk},~\cite{schmidli:2018risk}). For example, in automobile insurance $Y$ may describe road conditions, weather conditions (foggy days, rainy days, \dots), traffic volume, and so on.
While in~\cite{liangbayraktar:optreins} the authors considered a Markov-modulated compound Poisson process with the (unobservable) stochastic factor described by a finite state Markov chain,  we consider a stochastic factor model where the  exogenous process follows a general diffusion. An additional feature is that the insurance and the reinsurance premia are not evaluated using premium calculation principles, contrary to the majority of the literature; moreover, they turn out to be stochastic processes depending on $Y$.
Furthermore, we highlight that under the most frequently used premium calculation principles (expected value and variance premium principles) some problems arise: firstly, the optimal reinsurance strategy turns out to be deterministic (this is a limiting factor because the main goal of our paper is to consider a stochastic factor model); secondly, the optimal reinsurance strategy does not explicitly depend on the claims intensity. In order to fix these problems, we will introduce a new premium calculation principle, which is called \textit{intensity-adjusted variance premium principle}.


Finally, the financial market is more general than those usually considered in the literature, since it is composed by a risk-free bond and a risky asset with Markovian rate of return and volatility. For instance, in~\cite{baiguo:optreins},~\cite{caowan:optreins},~\cite{zhangyu:optreins} and~\cite{liangbayraktar:optreins} the authors used a geometric Brownian model, in~\cite{gu_yang:CEVoptreins} and~\cite{shengrongzhao:optreins} a CEV model. Nevertheless, some authors considered other general models: in~\cite{irgens_paulsen:optcontrol} and~\cite{lietal:robustXL} the risky asset follows a jump-diffusion process with constant parameters, in~\cite{liangetal:optreins} the instantaneous rate of investment return follows an Ornstein-Uhlenbeck process, in~\cite{zhuetal:reins_defaultable} the authors used the Heston model, in~\cite{xu_zang_yao:markov_mod_fin} the authors introduced a Markov-modulated model for the financial market. However, in these papers the authors considered the classical risk model with constant intensity for the claims arrival process.

Using the classical stochastic control approach based on the Hamilton-Jacobi-Bellman equation we characterize the optimal strategy and provide a verification result for the value function via classical solutions of two backward partial differential equations (see Theorem~\ref{theorem:verification}). Moreover we provide a class of sufficient conditions for existence and uniqueness of classical  solutions to the PDEs involved (see Theorems~\ref{theorem:g_pde} and~\ref{theorem:f_pde}). Results under various premium calculation principles are discussed, including the \textit{intensity-adjusted variance premium principle}. Finally, numerical simulations are performed to obtain sensitivity analyses of the optimal strategies.

The paper is organized as follows: in Section~\ref{section:formulation} we formulate the main assumptions and describe the optimization problem; Section~\ref{section:HJB} contains the derivation of the Hamilton-Jacobi-Bellman equation. In Section~\ref{section:optimal_reinsurance_strategy} we characterize the optimal reinsurance strategy, discussing in Subsections~\ref{section:premium_calc_principles} and~\ref{section:modified_variance} how the general results apply to special premium calculation principles (expected value, variance premium and intensity-adjusted variance principles). In Section~\ref{section:optimal_investment} we provide the optimal investment strategy. Section~\ref{section:verification_theorem} contains the Verification Theorem. In Section~\ref{section:simulation} we illustrate some numerical results and sensitivity analyses. In Section~\ref{section:PDE} existence and uniqueness theorems are discussed for the PDEs involved in the problem. Finally, in Appendix~\ref{section:appendix} the reader can find some proofs of secondary results.

\section{Problem formulation}
\label{section:formulation}

Assume that $(\Omega,\mathcal{F},\mathbb{P},\{\mathcal{F}_t\})$ is a complete probability space endowed with a filtration $\{\mathcal{F}_t\}_{t\in [0,T]}$, shortly denoted with $\{\mathcal{F}_t\}$, satisfying the usual conditions. We introduce the stochastic factor $Y=\{Y_t\}_{t\in [0,T]}$ as the solution of the following SDE:
\begin{equation}
\label{eqn:stochasticfactor}
dY_t = b(t,Y_t)\,dt + \gamma(t,Y_t)\,dW^{(Y)}_t \qquad \qquad Y_0\in\mathbb{R}
\end{equation}
where $\{W^{(Y)}_t\}_{t\in [0,T]}$ is a standard Brownian motion on $(\Omega,\mathcal{F},\mathbb{P},\{\mathcal{F}_t\})$. 
This stochastic factor represents any environmental alteration reflecting on risk fluctuations. For instance, as suggested by Grandell, J. (see~\cite{grandell:risk}, Chapter 2), in automobile insurance $Y$ may describe road conditions, weather conditions (foggy days, rainy days, \dots), traffic volume, and so on.

We suppose that there exists a unique strong solution to~\eqref{eqn:stochasticfactor} such that
\begin{gather}
\label{eqn:solutionY}
\mathbb{E}\biggl[\int_0^T\abs{b(t,Y_t)}\,dt+\int_0^T \gamma(t,Y_t)^2\,dt\biggr]<\infty \\
\label{eqn:solutionY2}
\sup_{t\in[0,T]}{\mathbb{E}[\abs{Y_t}^2]}<\infty
\end{gather}
(for instance, it is true if the coefficients of the SDE~\eqref{eqn:stochasticfactor} satisfy the classical Lipschitz and sub-linear growth conditions, see~\cite{gihmanskorohod:sde}) and denote by $\mathcal{L}^Y$ its infinitesimal generator:
\[
\mathcal{L}^Yf(t,y) = b(t,y) \frac{\partial{f}}{\partial{y}}(t,y)+\frac{1}{2}\gamma(t,y)^2\frac{\partial^2{f}}{\partial{y^2}}(t,y)
\qquad f\in\mathcal{C}^{1,2}((0,T)\times\mathbb{R}).
\]
Let us introduce a strictly positive measurable function $\lambda(t,y):[0,T]\times\mathbb{R}\to(0,+\infty)$ and define the process $\{\lambda_t\doteq\lambda(t,Y_t)\}_{t\in [0,T]}$ for all $t\in[0,T]$. Under the hypothesis that
\begin{equation}
\label{eqn:intensity_integr}
\mathbb{E}\biggl[\int_0^T \lambda_u\,du \biggr] <\infty
\end{equation}
we denote by $\{N_t\}_{t\in [0,T]}$ the claims arrival process, which is a conditional Poisson process having $\{\lambda_t\}_{t\in [0,T]}$ as intensity.
More precisely, we have that for all $0\le s\le t\le T$ and $k=0,1,\dots$
\[
\mathbb{P}[N_t-N_s=k\mid\mathcal{F}^Y_T\lor\mathcal{F}_s] = \frac{\bigl(\int_s^t \lambda_u\,du\bigr)^k}{k!}e^{-\int_s^t \lambda_u\,du},
\]
where $\{\mathcal{F}^Y_t\}_{t\in [0,T]}$ denotes the filtration generated by $Y$.
Then it is easy to show that
\[
N_t - \int_0^t \lambda_s\,ds
\]
is an $\{\mathcal{F}_t\}$-martingale%
\footnote{See e.g. \cite[II]{bremaud:pointproc}}.\\
Now we define the cumulative claims up to time $t$ as follows:
\[
C_t =\sum_{i=1}^{N_t}{Z_i} \qquad t\in[0,T],
\]
where the sequence of i.i.d. strictly positive $\mathcal{F}_0$-random variables $\{Z_i\}_{i=1,\dots}$ represents the amount of the claims. In the sequel we will assume that all the $\{Z_i\}_{i=1,\dots}$ are distributed like a r.v. $Z$, independent on $\{N_t\}_{t\in [0,T]}$ and $\{Y_t\}_{t\in [0,T]}$, with distribution function $F_Z(dz)$ such that $F_Z(z)=1\quad\forall z\ge D$, with $D\in\mathbb{R}^+$ (eventually $D=+\infty$). Moreover, $Z$ satisfies some suitable integrability conditions (see~\eqref{eqn:eZ_finite} below).\\
Consider the random measure associated with the marked point process $\{C_t\}_{t\in [0,T]}$ defined as follows

\begin{align}
m(dt,dz) &= \sum_{\substack{t\in[0,T]:\\ \Delta C_t\neq0}}{\delta_{(t,\Delta C_t)}(dt,dz)} \notag \\
\label{eqn:random_measure}
&= \sum_{n\ge1}{\delta_{(T_n,Z_n)}(dt,dz)}\mathbbm{1}_{\{T_n\le T\}},
\end{align}
where $\{T_n\}_{n=1,\dots}$ denotes the sequence of jump times of $\{N_t\}_{t\in [0,T]}$,
then the process $\{C_t\}_{t\in [0,T]}$  satisfies
\begin{equation}
\label{eqn:cumulative_claims}
C_t = \int_0^t \int_0^D z m(ds,dz).
\end{equation}

The following Lemma will be useful in the sequel.
\begin{lemma}
\label{prop:random_measure}
The random measure $m(dt,dz)$ given in~\eqref{eqn:random_measure} has dual predictable projection $\nu$ given by the following:
\begin{equation}
\label{eqn:dual_projection}
\nu(dt,dz) = dF_Z(z)\lambda_t\,dt
\end{equation}
i.e. for every nonnegative, $\{\mathcal{F}_t\}$-predictable and $[0,D]$-indexed process $\{H(t,z)\}_{t\in[0,T]}$
\[
\mathbb{E}\biggl[\int_0^T\int_0^DH(t,z)\,m(dt,dz)\biggr]=\mathbb{E}\biggl[\int_0^T\int_0^D H(t,z)\,dF_Z(z)\lambda_t\,dt\biggr].
\]
\end{lemma}
\begin{proof}
See Appendix A.
\end{proof}

\begin{remark}
\label{remark:random_measure}
Let us observe that for any $\{\mathcal{F}_t\}$-predictable and $[0,D]$-indexed process $\{H(t,z)\}_{t\in[0,T]}$ such that
\[
\mathbb{E}\biggl[\int_0^T\int_0^D \abs{H(t,z)}\,dF_Z(z)\lambda_t\,dt\biggr]<\infty
\]
the process
\[
M_t = \int_0^t\int_0^DH(s,z)\,\bigl(m(ds,dz)-dF_Z(z)\lambda_s\,ds\bigr) \qquad t\in[0,T]
\]
turns out to be an $\{\mathcal{F}_t\}$-martingale. If in addition
\[
\mathbb{E}\biggl[\int_0^T\int_0^D \abs{H(t,z)}^2\,dF_Z(z)\lambda_t\,dt\biggr]<\infty,
\]
then $\{M_t\}_{t\in [0,T]}$ is a square integrable $\{\mathcal{F}_t\}$-martingale and
\[
\mathbb{E}[M_t^2] = \mathbb{E}\biggl[\int_0^t\int_0^D \abs{H(t,z)}^2\,dF_Z(z)\lambda_t\,dt\biggr] \qquad\forall t\in[0,T].
\]
Moreover, the predictable covariation process of $\{M_t\}_{t\in[0,T]}$ is given by
\[
\langle M\rangle_t = \int_0^T\int_0^D \abs{H(t,z)}^2\,dF_Z(z)\lambda_t\,dt
\]
that is $\{M_t^2 - \langle M\rangle_t\}_{t\in[0,T]}$ is an $\{\mathcal{F}_t\}$-martingale%
\footnote{For these results and other related topics see e.g.~\cite{bass:2004}.}.
\end{remark}

\begin{remark}
\label{remark:filtration_G}
Let $\{\mathcal{G}_t\}_{t\in[0,T]}$ be the filtration defined by $\mathcal{G}_t=\mathcal{F}_t\lor\mathcal{F}^Y_T$. Then $m(dt,dz)$ defined in~\eqref{eqn:random_measure} has $\{\mathcal{G}_t\}$-dual predictable projection $\nu$ given in~\eqref{eqn:dual_projection}. In fact, first observe that $\{\lambda_t\}_{t\in [0,T]}$ is $\{\mathcal{F}_t\}$-adapted by definition, hence it is $\{\mathcal{G}_t\}$-adapted. Now notice that $\{\lambda_t\}$ is the $\{\mathcal{G}_t\}$-intensity of $\{N_t\}_{t\in[0,T]}$ because for any $0\le s\le t\le T$
\begin{align*}
\mathbb{E}[N_t\mid\mathcal{G}_s] &= N_s + \mathbb{E}[N_t-N_s\mid\mathcal{G}_s]\\
&= N_s + \sum_{k\ge1}{k\frac{\bigl(\int_s^t \lambda_u\,du\bigr)^k}{k!}e^{-\int_s^t \lambda_u\,du}}\\
&= N_s + \int_s^t \lambda_u\,du
\end{align*}
and this implies that
\[
\mathbb{E}[N_t-\int_0^t \lambda_u\,du\mid\mathcal{G}_s] = N_s-\int_0^s \lambda_u\,du.
\]
Then our statement follows by the proof of Lemma~\ref{prop:random_measure} (see Appendix A) by replacing $\{\mathcal{F}_t\}$-predictable and $[0,D]$-indexed processes with $\{\mathcal{G}_t\}$-predictable and $[0,D]$-indexed processes.
\end{remark}

In this framework we suppose that the gross risk premium rate is affected by the stochastic factor, i.e. we describe the insurance premium as a stochastic process $\{c_t\doteq c(t,Y_t)\}_{t\in [0,T]}$, where $c:[0,T]\times\mathbb{R}\to(0,+\infty)$ is a nonnegative measurable function such that
\begin{equation}
\label{eqn:cpremium_int}
\mathbb{E}\biggl[\int_0^T c(t,Y_t)\,dt\biggr]<\infty.
\end{equation}
The insurance company can continuously purchase a proportional reinsurance contract, transferring at each time $t\in[0,T]$ a percentage $u_t$ of its own risks to the reinsurer, who receives a reinsurance premium $q_t$ given by the definition below.
\begin{definition}(Proportional reinsurance premium)
\label{def:reinsurance_premium}
Let us define a function $q(t,y,u):[0,T]\times\mathbb{R}\times[0,1]\to[0,+\infty)$, continuous w.r.t. the triple $(t,y,u)$, having continuous partial derivatives $\frac{\partial q(t,y,u)}{\partial u},\frac{\partial^2 q(t,y,u)}{\partial u^2}$ in $u\in[0,1]$ and such that
\begin{enumerate}
\item $q(t,y,0)=0$ for all $(t,y)\in[0,T]\times\mathbb{R}$, because a null protection is not expensive;
\item $\frac{\partial q(t,y,u)}{\partial u}\ge0$ for all $(t,y,u)\in[0,T]\times\mathbb{R}\times[0,1]$, since the premium is increasing with respect to the protection;
\item $q(t,y,1)>c(t,y)$ for all $(t,y)\in[0,T]\times\mathbb{R}$, because the cedant is not allowed to gain a profit without risk.
\end{enumerate}
In the rest of the paper $\frac{\partial q(t,y,0)}{\partial u}$ and $\frac{\partial q(t,y,1)}{\partial u}$ should be intended as right and left derivatives, respectively.
Moreover, we assume the following integrability condition:
\begin{equation}
\label{eqn:qpremium_int}
\mathbb{E}\biggl[\int_0^T q(t,Y_t,u)\,dt\biggr]<\infty \quad \forall u\in[0,1].
\end{equation}
Then the reinsurance premium associated with a reinsurance strategy $\{u_t\}_{t\in [0,T]}$ (which is the protection level chosen by the insurer) is defined as $\{q_t\doteq q(t,Y_t,u_t)\}_{t\in [0,T]}$.
\end{definition}
In addition, we will use the hypothesis that the insurance gross premium and the reinsurance premium will never diverge too much (being approximately influenced by the stochastic factor in the same way), that is there exists a positive constant $K$ such that
\begin{equation}
\label{eqn:premia_bounded_prop}
\abs{q(t,Y_t,u)-c(t,Y_t)}\le K \quad \mathbb{P}\text{-a.s.} \quad \forall t\in[0,T],u\in[0,1].
\end{equation}
Under these hypotheses the surplus (or reserve) process associated with a given reinsurance strategy $\{u_t\}_{t\in [0,T]}$ is described by the following SDE:
\begin{align}
\label{eqn:surplus_process}
dR^u_t &= \biggl[c(t,Y_t)-q(t,Y_t,u_t)\biggl]dt - (1-u_t)dC_t \notag\\
&= \biggl[c(t,Y_t)-q(t,Y_t,u_t)\biggl]dt - \int_0^D(1-u_t)z\,m(dt,dz) \qquad R^u_0= R_0\in\mathbb{R}^+
\end{align}
Let us observe that by Remark  \ref{remark:random_measure}, since
\[
\mathbb{E}\biggl[\int_0^T\int_0^D{u_rz\lambda_r\,dF_Z(z)\,dr}\biggr] \le \mathbb{E}[Z]\mathbb{E}\biggl[\int_0^T{\lambda_r\,dr}\biggr]<\infty,
\]
the process $\int_0^t\int_0^D{(1-u_s)z(m(ds,dz)-\lambda_s\,dF_Z(z)\,ds)}$ turns out to be an $\{\mathcal{F}_t\}$-martingale.\\

Furthermore, we allow the insurer to invest its surplus in a financial market consisting of a risk-free bond $\{B_t\}_{t\in [0,T]}$ and a risky asset $\{P_t\}_{t\in [0,T]}$, whose dynamics on $(\Omega,\mathcal{F},\mathbb{P},\{\mathcal{F}_t\})$ are, respectively,
\begin{equation}
dB_t = RB_t\,dt \qquad B_0=1
\end{equation}
with a fixed $R>0$, and
\begin{equation}
\label{eqn:priceprocess}
dP_t = P_t\biggl[\mu(t,P_t)\,dt + \sigma(t,P_t)\,dW^{(P)}_t\biggr] \qquad P_0>0
\end{equation}
where $\{W^{(P)}_t\}_{t\in[0,T]}$ is a standard Brownian motion independent of $\{W^{(Y)}\}_{t\in[0,T]}$ and the random measure $m(dt,dz)$%
\footnote{This is a classical assumption which implies that the financial market is independent on the insurance market.}.
Let us assume that there exists a unique strong solution to~\eqref{eqn:priceprocess} such that
\begin{gather}
\label{eqn:solutionP}
\mathbb{E}\biggl[\int_0^T\abs{P_t\mu(t,P_t)}\,dt+\int_0^T P_t^2\sigma(t,P_t)^2\,dt\biggr]<\infty \\
\label{eqn:solutionP2}
\sup_{t\in[0,T]}{\mathbb{E}[P_t^2]}<\infty
\end{gather}
(for example, it is true if the coefficients of the SDE~\eqref{eqn:priceprocess} satisfy the classical Lipschitz and sub-linear growth conditions, see~\cite{gihmanskorohod:sde}).
Furthermore, we assume the Novikov condition:
\begin{equation}
\label{eqn:novikov}
\mathbb{E}\biggl[e^{\frac{1}{2}\int_0^T \abs{\frac{\mu(t,P_t)-R}{\sigma(t,P_t)}}^2\,dt} \biggr]<\infty,
\end{equation}
which implies the existence of a risk-neutral measure for $\{P_t\}_{t\in [0,T]}$ and ensures that the financial market does not admit arbitrage.

We will denote with $w_t$ the total amount invested in the risky asset at time $t\in[0,T]$, so that $X_t-w_t$ will be the capital invested in the risk-free asset (now $X_t$ indicates the total wealth, but it will be defined more accurately below, see equation~\eqref{eqn:wealth_sol}). We also allow the insurer to short-sell and to borrow/lend any infinitesimal amount, so that $w_t\in\mathbb{R}$.\\
Finally, we only consider self-financing strategies: the insurer company only invests the surplus obtained with the core business, neither subtracting anything from the gains, nor adding something from another business.\\

The insurer's wealth $\{X^\alpha_t\}_{t\in [0,T]}$ associated with a given strategy $\alpha_t=(u_t,w_t)$ is described by the following SDE:
\begin{align}
dX^\alpha_t &= dR^u_t+ w_t \frac{dP_t}{P_t} + \bigl(X^\alpha_t - w_t\bigr)\frac{dB_t}{B_t} \notag\\
&= \biggl[c(t,Y_t)-q(t,Y_t,u_t)\biggl]dt  + w_t\biggl[\mu(t,P_t)\,dt + \sigma(t,P_t)\,dW^{(P)}_t\biggr] \notag\\
\label{eqn:wealth_proc}
&+ \bigl(X^\alpha_t - w_t\bigr)R\,dt - \int_0^D(1-u_t)z\,m(dt,dz)
\end{align}
with $X^\alpha_0=R_0\in\mathbb{R}^+$. Remember that $\{u_t\}_{t\in [0,T]}$ and $\{w_t\}_{t\in [0,T]}$ are, respectively, the proportion of reinsured claims and the total amount invested in the risky asset $\{P_t\}_{t\in [0,T]}$.\\
\begin{remark}
It can be verified that the solution of the SDE~\eqref{eqn:wealth_proc} is given by the following:
\begin{multline}
\label{eqn:wealth_sol}
X^\alpha_t = X^\alpha_0e^{Rt} + \int_0^t e^{R(t-r)}\bigl[c(r,Y_r)-q(r,Y_r,u_r)\bigr]\,dr
+\int_0^t{e^{R(t-r)}w_r[\mu(r,P_r)-R]\,dr} \\
+ \int_0^t{e^{R(t-r)}w_r\sigma(r,P_r)\,dW^{(P)}_r}
-\int_0^t\int_0^D e^{R(t-r)}(1-u_r)z\,m(dr,dz).
\end{multline}
\end{remark}

Now we are ready to formulate the optimization problem of an insurance company which subscribes a proportional reinsurance contract and invests its surplus in a financial market according with a strategy $\{\alpha_t=(u_t,w_t)\}_{t\in[0,T]}$ in order to maximize the expected utility of its terminal wealth:
\[
\sup_{\alpha\in\mathcal{U}}{\mathbb{E}\bigl[U(X_T^\alpha)\bigr]}
\]
where $\mathcal{U}$ denotes a suitable class of admissible controls defined below (see Definition~\ref{def_U}) and $U:\mathbb{R}\to[0,+\infty)$ is the utility function representing the insurer preferences. We focus on CARA (\textit{Constant Absolute Risk Aversion}) utility functions, whose general expression is given by
\[
U(x) = 1-e^{-\eta x} \qquad x\in\mathbb{R}
\]
where $\eta>0$ is the risk-aversion parameter. This utility function is highly relevant in economic science and in particular in insurance theory, in fact it is commonly used for reinsurance problems (e.g. see ~\cite{baiguo:optreins},~\cite{caowan:optreins},~\cite{shengrongzhao:optreins}, and many others).\\
Using the dynamic programming principle we will consider a dynamic problem which consists in finding the optimal strategy $\alpha_s$, for $s\in[t,T]$, for the following optimization problem given the information available at the time $t\in[0,T]$:
\[
\sup_{\alpha\in\mathcal{U}_t}{\mathbb{E}\biggl[U(X_{t,x}^\alpha(T))\mid \mathcal{F}_t\biggr]} \qquad t\in[0,T]
\]
where $\mathcal{U}_t$ denotes the class of admissible controls in the time interval $[t,T]$ (see Definition~\ref{def_U} below). Here $\{X_{t,x}^\alpha(s)\}_{s\in[t,T]}$ denotes the solution to equation~\eqref{eqn:wealth_proc} with initial condition $X_t^\alpha=x$.\\
For the sake of simplicity, we will reduce ourselves studying the function $-e^{-\eta x}$. Another possible choice is to study the corresponding minimizing problem for the function $e^{-\eta x}$, but the first choice is usually preferred in the literature.

\begin{definition}
\label{def_U}
We will denote with $\mathcal{U}$ the set of all admissible strategies, which are all the $\{\mathcal{F}_t\}$-predictable processes $\alpha_t=(u_t,w_t)$, $t\in[0,T]$, with values in $[0,1]\times\mathbb{R}$, such that
\begin{gather*}
\mathbb{E}\biggl[\int_0^T{\abs{w_r}\abs{\mu(r,P_r)-R}\,dr}\biggr] <\infty , \quad
\mathbb{E}\biggl[\int_0^T{w_r^2\sigma(r,P_r)^2\,dr}\biggr] <\infty.
\end{gather*}

When we want to restrict the controls to the time interval $[t,T]$, we will use the notation $\mathcal{U}_t$.
\end{definition}

From now on we assume the following assumptions fulfilled.
\begin{assumption}
\label{assumption:EZ_finite}
\begin{gather}
\label{eqn:eZ_finite}
\mathbb{E}[e^{\eta Ze^{RT}}]<\infty, \qquad \mathbb{E}[Ze^{\eta Ze^{RT}}]<\infty
\qquad \mathbb{E}[Z^2e^{\eta Ze^{RT}}]<\infty \\
\label{lambda_int}
\mathbb{E}\bigl[e^{(\mathbb{E}[e^{\eta e^{RT}Z}]-1)\int_t^T \lambda_s\,ds}\mid\mathcal{F}_t\bigr]<\infty
\qquad\Braket{\mathbb{P}=1} \quad \forall t\in[0,T].
\end{gather}
\end{assumption}

\begin{prop}
\label{nullcontrol_admissible}
Under the Assumption~\ref{assumption:EZ_finite} the control  $(0,0)$ is admissible and such that
\begin{equation*}
\mathbb{E}[e^{-\eta X_{t,x}^{(0,0)} (T)}\mid\mathcal{F}_t]<\infty \qquad \Braket{\mathbb{P}=1} \qquad \forall (t,x)\in[0,T]\times\mathbb{R}.
\end{equation*}
\end{prop}
\begin{proof}
See Appendix~\ref{section:appendix}.
\end{proof}

\begin{remark}
Let us observe that Proposition \ref{nullcontrol_admissible} implies that 
\[
\esssup_{\alpha\in\mathcal{U}_t}{\mathbb{E}\biggl[U(X_{t,x}^\alpha(T))\mid \mathcal{F}_t\biggr]} > - \infty
\qquad \Braket{\mathbb{P}=1} \qquad t\in[0,T]
\]
and as a consequence that
\[
\sup_{\alpha\in\mathcal{U}}{\mathbb{E}\bigl[U(X_T^\alpha)\bigr]} > - \infty.
\]
\end{remark}

In order to solve this dynamic problem we introduce the \textit{value function} associated with it
\begin{equation}
v(t,x,y,p)= \sup_{\alpha\in\mathcal{U}_t}{\mathbb{E}\biggl[-e^{-\eta X_{t,x}^\alpha(T)}\mid Y_t=y,P_t=p\biggr]}
\end{equation}
where the function $v:V\to\mathbb{R}$ is defined in the domain
\[
V\doteq[0,T]\times\mathbb{R}^2\times(0,+\infty).
\]

The following Lemma  gives sufficient conditions to extend Proposition \ref{nullcontrol_admissible} to all constant strategies.

\begin{lemma}
\label{lemma:constant_strategies_prop}
Under the Assumption~\ref{assumption:EZ_finite}, let us suppose $\sigma(t,p)$ and $\mu(t,p)$ are bounded for all $(t,p)\in[0,T]\times(0,+\infty)$. Then we have that all constant strategies $\alpha_t=(u,w)$ with $u\in[0,1]$ and $w\in\mathbb{R}$ are admissible and such that
\begin{equation*}
\mathbb{E}[e^{-\eta X_{t,x}^{\alpha} (T)}\mid\mathcal{F}_t]<\infty \qquad \Braket{\mathbb{P}=1} \qquad \forall (t,x)\in[0,T]\times\mathbb{R}.
\end{equation*}
\end{lemma}
\begin{proof}
See Appendix~\ref{section:appendix}.
\end{proof}


\section{Hamilton-Jacobi-Bellman equation}
\label{section:HJB}

Let us consider the Hamilton-Jacobi-Bellman equation that the value function is expected to solve if sufficiently regular 
\begin{equation}
\label{eqn:HJBformulation_prop}
\left\{
\begin{aligned}
&\sup_{(u,w)\in [0,1]\times\mathbb{R}}{\mathcal{L}^\alpha v(t,x,y,p)} = 0\\
& v(T,x,y,p)= -e^{-\eta x} \qquad \forall(y,p)\in\mathbb{R}\times(0,+\infty)
\end{aligned}
\right.
\end{equation}
where $\mathcal{L}^\alpha$ denotes the Markov generator of the triple $(X^\alpha_t,Y_t,P_t)$ associated with a constant control $\alpha=(u,w)$. 
In what follows, we denote by $\mathcal{C}^{1,2}_b$ all bounded functions $f(t,x_1,\dots,x_n)$, with $n\ge1$, with bounded first order derivatives $\frac{\partial{f}}{\partial{t}},\frac{\partial{f}}{\partial{x_1}},\dots,\frac{\partial{f}}{\partial{x_n}}$ and bounded second order derivatives w.r.t. the spatial variables $\frac{\partial^2{f}}{\partial{x_1^2}},\dots,\frac{\partial{f}}{\partial{x_n^2}}$.

\begin{lemma}
\label{lemmagenerator}
Let $f:V\to\mathbb{R}$ be a function in $\mathcal{C}^{1,2}_b$. Then the Markov generator of the stochastic process $(X^\alpha_t,Y_t,P_t)$ for all constant strategies $\alpha=(u,w)\in[0,1]\times\mathbb{R}$ is given by the following expression:
\begin{multline}
\label{eqn:generator}
\mathcal{L}^\alpha f(t,x,y,p) = \frac{\partial{f}}{\partial{t}}(t,x,y,p)
+ \frac{\partial{f}}{\partial{x}}(t,x,y,p) \bigl[Rx+c(t,y)-q(t,y,u)+w(\mu(t,p)-R)\bigr] \\
+ \frac{1}{2}w^2\sigma(t,p)^2\frac{\partial^2{f}}{\partial{x^2}}(t,x,y,p)
+ b(t,y)\frac{\partial{f}}{\partial{y}}(t,x,y,p)+ \frac{1}{2}\gamma(t,y)^2\frac{\partial^2{f}}{\partial{y^2}}(t,x,y,p) \\
+ p\mu(t,p)\frac{\partial{f}}{\partial{p}}(t,x,y,p)
+ \frac{1}{2}p^2\sigma(t,p)^2\frac{\partial^2{f}}{\partial{p^2}}(t,x,y,p)
+ w\sigma(t,p)^2p \frac{\partial^2{f}}{\partial{x}\partial{p}}(t,x,y,p) \\
+ \int_0^D{\biggl[f(t,x-(1-u)z,y,p)-f(t,x,y,p)\biggr]\lambda(t,y)\,dF_Z(z)}.
\end{multline}
\end{lemma}
\begin{proof}
See Appendix A.
\end{proof}

Now let us introduce the following ansatz:
\[
v(t,x,y,p) =  -e^{-\eta xe^{R(T-t)}}\varphi(t,y,p)
\]
where $\varphi$ does not depend on $x$ and it is a positive function%
\footnote{Intuitively, we note that $X_{t,x}^\alpha(T)=X_{t,0}^\alpha(T)+xe^{R(T-t)}$ and we use the exponential form of the function $v$.}.
Then the original HJB problem given in~\eqref{eqn:HJBformulation_prop} reduces to the simpler one given by
\begin{multline}
\label{eqn:hjb_prop}
-\frac{\partial{\varphi}}{\partial{t}}(t,y,p) - b(t,y)\frac{\partial{\varphi}}{\partial{y}}(t,y,p) - \frac{1}{2}\gamma(t,y)^2 \frac{\partial^2{\varphi}}{\partial{y^2}}(t,y,p)
+ \eta e^{R(T-t)}c(t,y)\varphi(t,y,p)\\
- p\mu(t,p)\frac{\partial{\varphi}}{\partial{p}}(t,y,p)
- \frac{1}{2}\sigma(t,p)^2p^2\frac{\partial^2{\varphi}}{\partial{p^2}}(t,y,p)\\
+\sup_{u\in[0,1]}{\Psi^u(t,y)}\varphi(t,y,p)+ \sup_{w\in\mathbb{R}}{\Psi^w(t,y,p)}= 0
\end{multline}
with final condition $\varphi(T,y,p) = 1$ for all $(y,p)\in\mathbb{R}\times(0,+\infty)$, defining
\begin{equation}
\label{eqn:Psi_u}
\Psi^u(t,y)\doteq - \eta e^{R(T-t)}q(t,y,u) + \lambda(t,y)\int_0^D{\biggl[1-e^{\eta (1-u)ze^{R(T-t)}}\biggr]\,dF_Z(z)}
\end{equation}
and
\begin{align}
\label{eqn:Psi_w}
\Psi^w(t,y,p)&\doteq \eta e^{R(T-t)}\biggl((\mu(t,p)-R)\varphi(t,y,p)+p\sigma(t,p)^2 \frac{\partial{\varphi}}{\partial{p}}(t,y,p)\biggr)w \notag\\
&- \frac{1}{2}\sigma(t,p)^2 \eta^2 e^{2R(T-t)}\varphi(t,y,p)w^2.
\end{align}

It should make it clear that we can split the optimal control research in two distinct problems: the optimization of $\Psi^u$ will give us the optimal level of reinsurance (see Section~\ref{section:optimal_reinsurance_strategy}), while working with $\Psi^w$ we will find the optimal investment policy (see Section~\ref{section:optimal_investment}).


\section{Optimal reinsurance strategy}
\label{section:optimal_reinsurance_strategy}

In this section we discuss the problem
\begin{equation}
\label{eqn:opt_reins_problem}
\sup_{u\in[0,1]}{\Psi^u(t,y)}, \quad (t,y) \in [0,T] \times \ \mathbb{R}
\end{equation}
with $\Psi^u(t,y)$ given in \eqref{eqn:Psi_u}.

First, let us observe that  $\Psi^u(t,y)$  is continuous w.r.t. $u\in[0,1]$,  for any $(t,y) \in [0,T] \times \ \mathbb{R}$ and admits continuous first and the second order derivatives  w.r.t. $u\in[0,1]$

\begin{align*}
\frac{\partial \Psi^u(t,y)}{\partial u} & = -\eta e^{R(T-t)}\biggl[\frac{\partial q(t,y,u)}{\partial u} - \lambda(t,y)\int_0^D{ze^{\eta (1-u)ze^{R(T-t)}}\,dF_Z(z)} \biggr] \\
\frac{\partial^2 \Psi^u(t,y)}{\partial u^2}& = -\eta e^{R(T-t)}\biggl[\frac{\partial^2 q(t,y,u)}{\partial u^2} + \eta e^{R(T-t)}\lambda(t,y)\int_0^D{z^2e^{\eta (1-u)ze^{R(T-t)}}\,dF_Z(z)}\biggr].
\end{align*}
Notice that these derivatives are well defined thanks to~\eqref{eqn:eZ_finite}.\\

Now we are ready for the main result of this section.
\begin{prop}
\label{prop:optimalreins}
Given $\Psi^u(t,y)$ in~\eqref{eqn:Psi_u}, suppose that 
\begin{equation}
\label{eqn:PsiUconcave}
-\frac{\partial^2 q(t,y,u)}{\partial u^2}<\eta e^{R(T-t)}\lambda(t,y)\mathbb{E}\biggl[Z^2e^{\eta(1-u)Ze^{R(T-t)}}\biggr]
\qquad \forall(t,y,u)\in[0,T]\times\mathbb{R}\times(0,1).
\end{equation}
Then there exists a unique measurable function $u^*(t,y)$ for all $(t,y)\in[0,T]\times\mathbb{R}$ solution to~\eqref{eqn:opt_reins_problem}. Moreover, it is given by 
\begin{equation}
\label{eqn:reins_optimal_strategy}
u^*(t,y)=
\begin{cases}
	0 & \text{$(t,y)\in A_0$}
	\\
	\hat{u}(t,y) & \text{$(t,y)\in \hat{A}$}
	\\
	1 & \text{$(t,y)\in A_1$}
\end{cases}
\end{equation}
where
\begin{align*}
A_0&\doteq\Set{(t,y)\in[0,T]\times\mathbb{R}\mid\lambda(t,y)\mathbb{E}[Ze^{\eta Ze^{R(T-t)}}]\le\frac{\partial q(t,y,0)}{\partial u}}\\
\hat{A}&\doteq\Set{(t,y)\in[0,T]\times\mathbb{R}\mid\lambda(t,y)\mathbb{E}[Ze^{\eta Ze^{R(T-t)}}]>\frac{\partial q(t,y,0)}{\partial u}, \frac{\partial q(t,y,1)}{\partial u} > \mathbb{E}[Z]\lambda(t,y)}\\
A_1&\doteq\Set{(t,y)\in[0,T]\times\mathbb{R}\mid\frac{\partial q(t,y,1)}{\partial u} \le \mathbb{E}[Z]\lambda(t,y)}
\end{align*}
and $\hat{u}(t,y)$ is the unique solution of the following equation:
\begin{equation}
\label{eqn:first_order_cond}
\frac{\partial q(t,y,u)}{\partial u} = \lambda(t,y)\int_0^D{ze^{\eta (1-u)ze^{R(T-t)}}\,dF_Z(z)}.
\end{equation}
\end{prop}

\begin{proof}
Since $\Psi^u(t,y)$ is continuous in $u\in[0,1]$ and $\frac{\partial^2 \Psi^u(t,y)}{\partial u^2}<0$ $\forall(t,y,u)\in[0,T]\times\mathbb{R}\times(0,1)$  by~\eqref{eqn:PsiUconcave},  $\Psi^u(t,y)$ is strictly concave in $u\in [0,1]$.
As a consequence there exists a unique maximizer $u^*(t,y)$ of~\eqref{eqn:opt_reins_problem}, whose measurability follows by classical selection theorems.\\
Observe that $A_0\cup\hat{A}\cup A_1=[0,T]\times\mathbb{R}$. In fact, let us define these subsets as follows:
\begin{align*}
A_0&=\Set{(t,y)\in[0,T]\times\mathbb{R}\mid \frac{\partial \Psi^0(t,y)}{\partial u}\le0}\\
\hat{A}&=\Set{(t,y)\in[0,T]\times\mathbb{R}\mid \frac{\partial \Psi^0(t,y)}{\partial u}>0, \frac{\partial \Psi^1(t,y)}{\partial u}<0}\\
A_!&=\Set{(t,y)\in[0,T]\times\mathbb{R}\mid \frac{\partial \Psi^1(t,y)}{\partial u}\ge0}.
\end{align*}
Now, being $\frac{\partial \Psi^u(t,y)}{\partial u}$ strictly decreasing in $u\in(0,1)$, for any $(t,y)\in \hat{A}\cup A_1$ we have that
\[
\frac{\partial \Psi^0(t,y)}{\partial u}>\frac{\partial \Psi^1(t,y)}{\partial u}\ge0 \Rightarrow (t,y)\not\in A_0
\]
which implies that
\begin{align*}
[0,T]\times\mathbb{R}\setminus A_0 &= \Set{(t,y)\in[0,T]\times\mathbb{R}\mid \frac{\partial \Psi^0(t,y)}{\partial u}>0}\\
&= \Set{(t,y)\in[0,T]\times\mathbb{R}\mid \frac{\partial \Psi^0(t,y)}{\partial u}>0, \frac{\partial \Psi^1(t,y)}{\partial u}<0}\\
&\dot\cup \Set{(t,y)\in[0,T]\times\mathbb{R}\mid \frac{\partial \Psi^0(t,y)}{\partial u}>0, \frac{\partial \Psi^1(t,y)}{\partial u}\ge0}\\
&= \hat{A}\dot\cup A_1.
\end{align*}
Moreover, since $\hat{A}\cap A_1=\emptyset$, then $A_0\dot\cup\hat{A}\dot\cup A_1=[0,T]\times\mathbb{R}$.\\

Let us recall that $\frac{\partial \Psi^u(t,y)}{\partial u}$ is continuous and strictly decreasing in $u \in [0,1]$, for any $\forall(t,y)\in[0,T]\times\mathbb{R}$.

If $(t,y) \in A_0$ then $\Psi^u(t,y)$ is strictly decreasing in $u \in [0,1]$, hence no reinsurance is chosen, i.e. $u^*(t,y)=0$.

If  $(t,y) \in\hat{A}$ then there exists  a unique $u^*(t,y)\in(0,1)$ such that $\frac{\partial \Psi^u(t,y)}{\partial u} = 0$,  and it is the unique solution to equation~\eqref{eqn:first_order_cond}.

Finally, if  $(t,y) \in A_1$ then $\Psi^u(t,y)$ is strictly increasing in $u \in [0,1]$, hence $u^*(t,y)=1$.\\
\end{proof}

\begin{remark}
We also observe for the sake of completeness that if $\lambda(t,y)$ had been vanished for some $(t,y)$, then $\frac{\partial \Psi^u(t,y)}{\partial u}$ would have become strictly negative for all $u$, and in this case $u^*(t,y)=0$. In fact, the case of $\lambda(t,y)=0$ corresponds to a degenerate situation: the risk premia are paid, but there is no "real" risk to be insured.
\end{remark}

From the economic point of view, we could say that if the reinsurance is not too much expensive (more precisely, if the price of an infinitesimal protection is below a certain dynamic threshold) and if full reinsurance is not optimal, then the optimal strategy is provided by~\eqref{eqn:first_order_cond}, i.e. by equating the marginal cost and the marginal gain; moreover, the following remark points out the relevance of the third case in~\eqref{eqn:reins_optimal_strategy}.
\begin{remark}
In the current literature full reinsurance is always considered sub-optimal, contrary to the result given by formula~\eqref{eqn:reins_optimal_strategy}. The main reason is that using premium calculation principles many authors force the reinsurance premium to have certain properties, such as the convexity with respect to the protection level. In fact, it can be shown that if the reinsurance premium $q(t,y,u)$ is convex w.r.t. $u$, full reinsurance is never optimal (see Remark~\ref{remark:convexpremium_fullreins}). Nevertheless, it is reasonable that the insurer firm could regard full reinsurance as convenient for a limited period and in some particular scenarios, because actually the objective is to maximize the expected utility of the wealth at the end of the period.\\
Moreover, from the reinsurer's point of view, there is no reason to prevent the insurer from buying a full protection, providing the cedant is ready to pay a fair price. At the same time, if the reinsurer is not able to sell a full reinsurance, then it is sufficient to choose $q(t,y,u)$ such that $A_1= \emptyset$.
\end{remark}

Now we provide some sufficient conditions in order to guarantee that condition~\eqref{eqn:PsiUconcave} is fulfilled. 

\begin{lemma}
\label{PsiUconcave}
Suppose that at least one of the following condition holds:
\begin{enumerate}
\item $\frac{\partial q(t,y,0)}{\partial u}=0$ for all $(t,y)\in[0,T]\times\mathbb{R}$;
\item $\frac{\partial^2 q(t,y,u)}{\partial u^2}\ge0$ for all $u\in(0,1)$ and $(t,y)\in[0,T]\times\mathbb{R}$;
\item $-\frac{\partial^2 q(t,y,u)}{\partial u^2}< \eta\lambda(t,y)\mathbb{E}[Z^2]$ for all $u\in(0,1)$ and $(t,y)\in[0,T]\times\mathbb{R}$.
\end{enumerate}
Then the inequality~\eqref{eqn:PsiUconcave} holds, which implies that the function $\Psi^u(t,y)$ is strictly concave in $u\in(0,1)$.
\end{lemma}
\begin{proof}
First, let us observe that $1\Rightarrow2\Rightarrow3$. In fact, by the fundamental theorem of calculus we have that
\[
\frac{\partial q(t,y,u)}{\partial u} = \frac{\partial q(t,y,0)}{\partial u} + \int_0^u{\frac{\partial^2 q(t,y,w)}{\partial w^2}dw}
\]
and, being $\frac{\partial q(t,y,u)}{\partial u}\ge0$, $\frac{\partial q(t,y,0)}{\partial u}=0$ implies that the integrand function must be nonnegative, that is $1\Rightarrow2$. The implication $2\Rightarrow3$ is trivial, being $\eta>0,\lambda(t,y)>0$. Now it is sufficient to show that $3$ implies~\eqref{eqn:PsiUconcave}; clearly
\begin{align*}
-\frac{\partial^2 q(t,y,u)}{\partial u^2} &< \eta\lambda(t,y)\mathbb{E}[Z^2]\\
&< \eta e^{R(T-t)}\lambda(t,y)\mathbb{E}\biggl[Z^2e^{\eta(1-u)Ze^{R(T-t)}}\biggr]
\end{align*}
and hence $\frac{\partial^2 \Psi^u(t,y)}{\partial u^2}<0$ for all $u\in(0,1)$, i.e.~\eqref{eqn:PsiUconcave} holds, which implies that $\Psi^u(t,y)$ is strictly concave in $u\in(0,1)$
\end{proof}

\begin{remark}
\label{remark:convexpremium_fullreins}
Under the hypotheses that $\frac{\partial^2 q(t,y,u)}{\partial u^2}\ge0$ and $c(t,y)>\mathbb{E}[Z]\lambda(t,y)$ for all $(t,y,u)\in[0,T]\times\mathbb{R}\times(0,1)$, the full reinsurance is never optimal. In fact, for any arbitrary couple $(t,y)$ we have that
\[
q(t,y,1)=q(t,y,0)+\int_0^1{\frac{\partial q(t,y,u)}{\partial u}\,du}.
\]
Being $q(t,y,0)=0$ and $q(t,y,1)>c(t,y)>\mathbb{E}[Z]\lambda(t,y)$ (because the reinsurance is not cheap and using the \textit{net-profit condition} for the insurance premium), we obtain that
\[
\int_0^1{\frac{\partial q(t,y,u)}{\partial u}\,du}>\mathbb{E}[Z]\lambda(t,y).
\]
Since $\frac{\partial q(t,y,u)}{\partial u}$ is continuous in $u\in[0,1]$ by hypothesis, from the mean value theorem for integrals we know that there exists $u_0\in(0,1)$ such that
\[
\frac{\partial q(t,y,u_0)}{\partial u} > \mathbb{E}[Z]\lambda(t,y).
\]
Under the hypothesis that $\frac{\partial^2 q(t,y,u)}{\partial u^2}\ge0$ for all $u\in(0,1)$, $\frac{\partial q(t,y,u)}{\partial u}$ is an increasing function of $u$, and this implies that
\[
\frac{\partial q(t,y,1)}{\partial u} \ge \frac{\partial q(t,y,u_0)}{\partial u} > \mathbb{E}[Z]\lambda(t,y).
\]
From this result we deduce that
\[
\frac{\partial \Psi^1(t,y)}{\partial u} = -\eta e^{R(T-t)}\biggl[\frac{\partial q(t,y,1)}{\partial u} - \mathbb{E}[Z]\lambda(t,y)\biggr] < 0, \quad (t,y)\in[0,T]\times\mathbb{R} 
\]
which implies that $A_1= \emptyset$, i.e. 
the full reinsurance is never optimal.
\end{remark}

Let us observe that the preceding Remark requires two special conditions. The first one concerns the concavity of the reinsurance premium and in Subsection~\ref{section:premium_calc_principles} we will show that it is fulfilled by the most famous premium calculation principles. The second hypothesis is the so called \textit{net-profit condition} (e.g. see~\cite{grandell:risk}) and it is usually assumed in insurance risk models to ensure that the expected gross risk premium covers the expected losses.

Now we investigate how Proposition~\ref{prop:optimalreins} applies to a special case.

\begin{example}{(Exponentially distributed claims)} \label{example:exponential}\\
Let $Z$ to be an exponential r.v. with parameter $\zeta>0$,  then 
for any fixed $(t,y)\in[0,T]\times\mathbb{R}$ equation~\eqref{eqn:first_order_cond} becomes
\begin{equation*}
\lambda(t,y) \int_0^\infty{ze^{\eta(1-u)ze^{R(T-t)}}\zeta e^{-\zeta z}\,dz} = \frac{\partial q(t,y,u)}{\partial u}.
\end{equation*}
Taking $k=\eta(1-u)e^{R(T-t)}-\zeta$ it can be written as
\[
\lambda(t,y) \int_0^\infty{ze^{kz}\zeta \,dz} = \frac{\partial q(t,y,u)}{\partial u}
\]
and requiring that
\begin{equation}
\label{eqn:k_cond}
\frac{\zeta}{\eta}>e^{RT}
\end{equation}
which implies that $k<0$, finally equation~\eqref{eqn:first_order_cond} reads as 
\begin{equation}
\label{eqn:exp_claims_equation_u}
\lambda(t,y)\, \frac{\zeta}{(\eta(1-u)e^{R(T-t)}-\zeta)^2} = \frac{\partial q(t,y,u)}{\partial u}.
\end{equation}
Summarizing, if $Z$ is an exponential r.v. with parameter $\zeta>\eta e^{RT}$, under the condition~\eqref{eqn:PsiUconcave} we have that expression~\eqref{eqn:reins_optimal_strategy} holds with
\begin{align*}
A_0&\doteq\Set{(t,y)\in[0,T]\times\mathbb{R}\mid\lambda(t,y)\frac{\zeta}{(\eta e^{R(T-t)}-\zeta)^2}\le\frac{\partial q(t,y,0)}{\partial u}}\\
\hat{A}&\doteq\Set{(t,y)\in[0,T]\times\mathbb{R}\mid\lambda(t,y)\frac{\zeta}{(\eta e^{R(T-t)}-\zeta)^2}>\frac{\partial q(t,y,0)}{\partial u}, \frac{\partial q(t,y,1)}{\partial u}>\frac{\lambda(t,y)}{\zeta}}\\
A_1&\doteq\Set{(t,y)\in[0,T]\times\mathbb{R}\mid\frac{\partial q(t,y,1)}{\partial u}\le\frac{\lambda(t,y)}{\zeta}}
\end{align*}
and with $\hat{u}(t,y)$ being the unique solution to equation~\eqref{eqn:exp_claims_equation_u}.
\end{example}


\subsection{Expected value and variance premium principles}
\label{section:premium_calc_principles}

Proposition~\ref{prop:optimalreins} clarifies that the optimal reinsurance strategy crucially depends on the reinsurance premium. In this subsection we specialize that result using two of the most famous premium calculation principles: the \textit{expected value principle} and the \textit{variance premium principle}. We will show that in both cases we loose the dependence of the optimal reinsurance strategy on the stochastic factor. Moreover, the optimal reinsurance strategy does not explicitly depend on the claims intensity. These will be our motivations for introducing the \textit{intensity-adjusted variance premium principle} in Subsection~\ref{section:modified_variance}.

\begin{lemma}
\label{lemma:expected_vp_prop}
Under the \textit{expected value principle}, i.e. if the reinsurance premium admits the following expression
\begin{equation}
\label{eqn:expected_value_prop}
q(t,y,u) = (1+\theta_r)\mathbb{E}[Z]\lambda(t,y)u \qquad \forall (t,y,u)\in[0,T]\times\mathbb{R}\times[0,1]
\end{equation}
for some constant $\theta_r>0$ (which is called the reinsurance safety loading), there exists a unique maximizer $u^*(t)$ for all $(t,y)\in[0,T]\times\mathbb{R}$ for the problem~\eqref{eqn:opt_reins_problem}.
In particular,
\begin{equation}
u^*(t)=
\begin{cases}
	0 & \text{$t\in A_0$}
	\\
	\hat{u}(t) & \text{$t\in [0,T]\setminus A_0$}
\end{cases}
\end{equation}
where
\[
A_0\doteq\Set{t\in[0,T]\mid\mathbb{E}[Ze^{\eta Ze^{R(T-t)}}]\le(1+\theta_r)\mathbb{E}[Z]}
\]
and $\hat{u}(t)$ is the unique solution to the following equation:
\begin{equation}
\label{eqn:first_order_cond_expvalueprinc}
(1+\theta_r)\mathbb{E}[Z] = \int_0^D{ze^{\eta (1-u)ze^{R(T-t)}}\,dF_Z(z)}.
\end{equation}
\end{lemma}
\begin{proof}
From~\eqref{eqn:expected_value_prop} we get
\[
\frac{\partial q(t,y,u)}{\partial u} = (1+\theta_r)\mathbb{E}[Z]\lambda(t,y), \frac{\partial^2 q(t,y,u)}{\partial u^2}=0 \qquad \forall u\in(0,1)
\]
which implies that $\Psi^u(t,y)$ is strictly concave in $u\in(0,1)$ thanks to Lemma~\ref{PsiUconcave}. Moreover, by the means of Remark~\ref{remark:convexpremium_fullreins} we know that the full reinsurance is always sub-optimal, in fact the set $A_1$ in Proposition~\ref{prop:optimalreins} is empty. Now we only have to apply Proposition~\ref{prop:optimalreins}.
\end{proof}
Note that we always have $\mathbb{E}[Ze^{\eta Ze^{R(T-t)}}]>\mathbb{E}[Z]$ for each $t\in[0,T]$, thus  $A_0$ could be an empty set when the reinsurer's safety loading is close to $0$.\\

\begin{example}{(Exponentially distributed claims under the expected value principle)} \\
Let us come back to example~\ref{example:exponential}. Under the \textit{expected value principle}~\eqref{eqn:expected_value_prop} the result for exponential claims is even more simplified, in fact we find the following explicit solution:
\begin{equation}
\label{eqn:exponential_expvalue_opt_control}
u^*(t)=
\begin{cases}
1-\frac{\zeta}{\eta}\biggl(1-\frac{1}{\sqrt{1+\theta_r}}\biggr)e^{-R(T-t)} & \text{$t\in[0,t_0\land T)$} \\
0 & \text{$t\in[t_0\land T,T]$}
\end{cases}
\end{equation}
where
\begin{equation}
\label{eqn:t_0_expvalue}
t_0 = T-\frac{1}{R}\log\biggl[\frac{\zeta}{\eta}\biggl(1-\frac{1}{\sqrt{1+\theta_r}}\biggr)\biggr].
\end{equation}
The expression for $t_0$ can be derived from the characterization of the set $[0,T]\times\mathbb{R}\setminus A_0$, which in this case reads as follows:
\[
\frac{\zeta-\sqrt{\frac{\zeta}{(1+\theta_r)\mathbb{E}[Z]}}}{\eta}<e^{R(T-t)}<\frac{\zeta+\sqrt{\frac{\zeta}{(1+\theta_r)\mathbb{E}[Z]}}}{\eta}
\]
where the second inequality is always fulfilled in view of~\eqref{eqn:k_cond}, hence we get $t_0$ only from the first inequality.
\end{example}

\begin{lemma}
\label{lemma:variance_vp_prop}
Under the \textit{variance premium principle}, i.e. if the reinsurance premium admits the following expression
\begin{equation}
\label{eqn:variance_premium_prop}
q(t,y,u) = \mathbb{E}[Z]\lambda(t,y)u + \theta_r\mathbb{E}[Z^2]\lambda(t,y)u^2
\end{equation}
for some constant reinsurance safety loading $\theta_r>0$, the optimization problem~\eqref{eqn:opt_reins_problem} admits a unique maximizer $u^*(t)\in(0,1)$ for all $(t,y)\in[0,T]\times\mathbb{R}$ , which is the solution to the following equation:
\begin{equation}
\label{eqn:first_order_cond_variance}
2\theta_r\mathbb{E}[Z^2]u = \int_0^D{ze^{\eta (1-u)ze^{R(T-t)}}\,dF_Z(z)} - \mathbb{E}[Z].
\end{equation}
\end{lemma}
\begin{proof}
Using the expression~\eqref{eqn:variance_premium_prop} we get that
\[
\frac{\partial q(t,y,u)}{\partial u} = \mathbb{E}[Z]\lambda(t,y) + 2\theta_r\mathbb{E}[Z^2]\lambda(t,y)u \qquad \forall u\in(0,1)
\]
and
\[
\frac{\partial^2 q(t,y,u)}{\partial u^2}=2\theta_r\mathbb{E}[Z^2]\lambda(t,y)>0 \qquad \forall u\in(0,1).
\]
By Lemma~\ref{PsiUconcave} $\Psi^u(t,y)$ is strictly concave w.r.t. $u$ and the full reinsurance is never optimal because of Remark~\ref{remark:convexpremium_fullreins}. Moreover, in order to apply Proposition~\ref{prop:optimalreins} we notice that
\[
\mathbb{E}[Ze^{\eta Ze^{R(T-t)}}]>\mathbb{E}[Z] \Rightarrow A_0=\emptyset
\]
thus the optimal strategy is unique and it belongs to $(0,1)$. In order to find such a solution, we turn the attention to the first order condition, which is exactly the equation~\eqref{eqn:first_order_cond_variance}.
\end{proof}
The same result was obtained in~\cite{liangbayraktar:optreins}, Lemma 3.1.

\begin{example}{(Exponentially distributed claims under the variance premium principle)}
\label{example:variance_premium_exponential}\\
Under the \textit{variance premium principle}~\eqref{eqn:variance_premium_prop}, suppose that the claims are exponentially distributed with parameter $\zeta>\eta e^{RT}$. Then  it is easy to show that the optimal strategy is given by  
\begin{equation}
\label{eqn:exponential_varianceprem_opt_control}
u^*(t)=1-\frac{\zeta}{\eta}\biggl(1-\sqrt{\frac{\zeta}{\zeta+4\theta_r}}\biggr)e^{-R(T-t)} \qquad t\in[0,T].
\end{equation}
\end{example}


\subsection{Intensity-adjusted variance premium principle}
\label{section:modified_variance}

We have shown that both the \textit{expected value principle} (see Lemma~\ref{lemma:expected_vp_prop}) and the \textit{variance premium principle} (see Lemma~\ref{lemma:variance_vp_prop}) lead us to deterministic optimal reinsurance strategies, which do not depend on the stochastic factor. This is a limiting factor, since the main objective of our paper is to solve the maximization problem under a stochastic factor model.\\
In addition, in both cases the optimal reinsurance strategy does not explicitly depend on the claims intensity. As a consequence, there is a paradox that we clarify with the following example.
Let us consider two identical insurers (i.e. with the same risk-aversion, time horizon, and so on) who work in the same insurance business line, for example in automobile insurance, but in two distinct territories with different riskiness. More precisely, let us assume that the two companies insure claims which have the same distribution $F_Z$ but occur with different probabilities. Hence it is a reasonable assumption that the claims arrival processes have two different intensities. Now let us suppose that both the insurers use Lemma~\ref{lemma:expected_vp_prop} (or Lemma~\ref{lemma:variance_vp_prop}) in order to solve the maximization problem~\eqref{eqn:opt_reins_problem}. Then they will obtain the same reinsurance strategy, but this is not what we expect.
Hence the optimal reinsurance strategy should explicitly depend on the claims intensity.\\
In order to fix these two problems, in this subsection we introduce a new premium calculation principle, which will be referred as the \textit{intensity-adjusted variance premium principle}.\\

Let us first formalize that there exists a special class of premium calculation principles that lead us to deterministic strategies which do not depend on the claims intensity.

\begin{remark}
For any reinsurance premium $\{q_t\}_{t\in[0,T]}$ admitting the following representation
\begin{equation}
\label{eqn:q_factorization}
q(t,y,u)=\lambda(t,y)Q(t,u)
\end{equation}
for a suitable%
\footnote{I.e. $Q$ is such that $q$ fulfills the Definition~\ref{def:reinsurance_premium}.}
function $Q:[0,T]\times[0,1]\to[0,+\infty)$, the optimal reinsurance strategy $u^*_t=u^*(t,Y_t)$ given in Proposition~\ref{prop:optimalreins} turns out to be deterministic. Moreover, it does not explicitly depend on the claims intensity.
For example, the \textit{expected value principle} and the \textit{variance premium principle} admit the factorization~\eqref{eqn:q_factorization} with, respectively, $Q(t,u)=(1+\theta_r)\mathbb{E}[Z]u$ and $Q(t,u)=\mathbb{E}[Z]u + \theta_r\mathbb{E}[Z^2]u^2$.
\end{remark}

Now the basic idea is to find a reinsurance premium $\{q_t\}_{t\in[0,T]}$ (see Definition~\ref{def:reinsurance_premium}) such that
\begin{equation}
\label{eqn:vairance_prem_condition}
\mathbb{E}\biggl[\int_0^t{q(s,Y_s,u_s)\,ds}\biggr] =  \mathbb{E}\biggl[\int_0^t u_s\,dC_s\biggr]
+\theta_r\var{\biggl[\int_0^t u_s\,dC_s\biggr]} \qquad\forall t\in[0,T]
\end{equation}
for a given reinsurance safety loading $\theta_r$ in order to dynamically satisfy the original formulation of the \textit{variance premium principle}%
\footnote{See e.g.~\cite{young:premium_princ}.}. For this purpose, we give the following result.

\begin{lemma}
\label{lemma:variance_modified_var}
For any $\{\mathcal{F}^Y_t\}_{t\in[0,T]}$-predictable reinsurance strategy $\{u_t\}_{t\in[0,T]}$ we have that for any $t\in[0,T]$
\begin{equation}
\label{eqn:variance_cumulativeclaims}
\var{\biggl[\int_0^t u_s\,dC_s\biggr]} = \mathbb{E}[Z^2]\mathbb{E}\biggl[\int_0^t{u_s^2\lambda_s\,ds}\biggr]+ \mathbb{E}[Z]^2\var{\biggl[\int_0^t{u_s\lambda_s\,ds}\biggr]}.
\end{equation}
\end{lemma}
\begin{proof}
Let us denote with $\{M^u_t\}_{t\in[0,T]}$ the following $\{\mathcal{F}_t\}$-martingale:
\[
M^u_t = \int_0^D\int_0^t{u_sz\bigl(m(ds,dz)-dF_Z(z)\lambda_s\,ds\bigr)}.
\]
Recalling that $\{C_t\}_{t\in[0,T]}$ is defined in~\eqref{eqn:cumulative_claims}, the variance of the reinsurer's cumulative losses at the time $t\in[0,T]$ is given by
\begin{align*}
\var{\biggl[\int_0^t u_s\,dC_s\biggr]}
&= \mathbb{E}\biggl[\biggl(\int_0^t u_s\,dC_s\biggr)^2\biggr]-\mathbb{E}\biggl[\int_0^t u_s\,dC_s\biggr]^2\\
&= \mathbb{E}\biggl[\abs{M^u_t}^2+\biggl(\int_0^t{u_s\lambda_s\mathbb{E}[Z]\,ds}\biggr)^2+2M^u_t\int_0^t{u_s\lambda_s\mathbb{E}[Z]\,ds}\biggr]-\mathbb{E}\biggl[\int_0^t u_s\,dC_s\biggr]^2.
\end{align*}
Denoting with $\langle M^u\rangle_t$ the predictable covariance process of $M^u_t$, using Remark~\ref{remark:random_measure}, we find that
\begin{align}
\var{\biggl[\int_0^t u_s\,dC_s\biggr]}
&= \mathbb{E}[\langle M^u\rangle_t] + \mathbb{E}[Z]^2\mathbb{E}\biggl[\biggl(\int_0^t{u_s\lambda_s\,ds}\biggr)^2\biggr]-\mathbb{E}[Z]^2\mathbb{E}\biggl[\int_0^t{u_s\lambda_s\,ds}\biggr]^2 \notag\\
&= \mathbb{E}[Z^2]\mathbb{E}\biggl[\int_0^t{u_s^2\lambda_s\,ds}\biggr]+ \mathbb{E}[Z]^2\var{\biggl[\int_0^t{u_s\lambda_s\,ds}\biggr]} \qquad\forall t\in[0,T].
\end{align}
Here we have used that $\mathbb{E}\biggl[M^u_t\int_0^t{u_s\lambda_s\mathbb{E}[Z]\,ds}\biggr]=0$. In fact we notice that
\begin{align*}
\mathbb{E}\biggl[M^u_t\int_0^t{u_s\lambda_s\mathbb{E}[Z]\,ds}\biggr]
&= \mathbb{E}\biggl[\mathbb{E}\biggl[M^u_t\int_0^t{u_s\lambda_s\mathbb{E}[Z]\,ds}\mid\mathcal{F}^Y_T\biggr]\biggr]\\
&= \mathbb{E}\biggl[\mathbb{E}\bigl[M^u_t\mid\mathcal{F}^Y_T\bigr]\int_0^t{u_s\lambda_s\mathbb{E}[Z]\,ds}\biggr]
\end{align*}
and being $\mathcal{G}_0=\mathcal{F}_0\lor\mathcal{F}^Y_T\supseteq\mathcal{F}^Y_T$ (see Remark~\ref{remark:filtration_G}) we have that
\[
\mathbb{E}\bigl[M^u_t\mid\mathcal{F}^Y_T\bigr]=\mathbb{E}\bigl[\mathbb{E}\bigl[M^u_t\mid\mathcal{G}_0\bigr]\mathcal{F}^Y_T\bigr]=\mathbb{E}\bigl[M^u_0\mid\mathcal{F}^Y_T\bigr]=0
\]
and the proof is complete.
\end{proof}

\begin{remark}
We highlight that Lemma~\ref{lemma:variance_modified_var} applies to $\{\mathcal{F}^Y_t\}_{t\in[0,T]}$-predictable reinsurance strategies, but this is not restrictive. In fact, from Lemma~\ref{prop:optimalreins} we know that the optimal strategy belongs to the class of $\{\mathcal{F}^Y_t\}_{t\in[0,T]}$-predictable processes.
\end{remark}

\begin{remark}
In the classical Cram\'er-Lundberg model, i.e. $\lambda(t,y)=\lambda$, for any deterministic strategy $u_t=u(t)$
\[
\var{\biggl[\int_0^t{u_s\lambda\,ds}\biggr]}=0,
\]
thus in this case we choose expression~\eqref{eqn:variance_premium_prop} and the equation~\eqref{eqn:vairance_prem_condition} is satisfied.\\
\end{remark}

Under any risk model with stochastic intensity the formula~\eqref{eqn:variance_premium_prop} neglects the term
\[
\mathbb{E}[Z]^2\var{\biggl[\int_0^t{u_s\lambda_s\,ds}\biggr]}
\]
in the equation~\eqref{eqn:variance_cumulativeclaims}.
In order to capture the effect of this term, we can find the following estimate:
\begin{align*}
\var{\biggl[\int_0^t{u_s\lambda_s\,ds}\biggr]} &\le \mathbb{E}\biggl[\biggl(\int_0^t{u_s\lambda_s\,ds}\biggr)^2\biggr]\\
&\le \mathbb{E}\biggl[T\int_0^t{u_s^2\lambda_s^2\,ds}\biggr].
\end{align*}
As a consequence, we can choose as premium calculation rule
\begin{equation}
\label{eqn:variance_modified_prop}
q(t,y,u) = \mathbb{E}[Z]\lambda(t,y)u + \theta_r\mathbb{E}[Z^2]\bigl[\lambda(t,y)+T\lambda(t,y)^2\bigr]u^2
\end{equation}
which will be called \textit{intensity-adjusted variance principle} in this work; using this formula, we ensure that
\[
\mathbb{E}\biggl[\int_0^t{q(s,Y_s,u_s)\,ds}\biggr] \ge  \mathbb{E}\biggl[\int_0^t u_s\,dC_s\biggr]
+\theta_r\var{\biggl[\int_0^t u_s\,dC_s\biggr]} \qquad \forall t\in[0,T]
\]
for all $\{\mathcal{F}^Y_t\}_{t\in[0,T]}$-predictable reinsurance strategies and for any arbitrary level of reinsurance safety loading $\theta_r>0$.\\

\begin{lemma}
\label{lemma:variance_modified_prop}
Under the \textit{intensity-adjusted variance premium principle}~\eqref{eqn:variance_modified_prop}, the optimization problem~\eqref{eqn:opt_reins_problem} admits a unique maximizer $u^*(t,y)\in(0,1)$ for all $(t,y)\in[0,T]\times\mathbb{R}$, which is the solution to the following equation:
\begin{equation}
\label{eqn:first_order_cond_modvariance}
2\theta_r\mathbb{E}[Z^2]\bigl[1+T\lambda(t,y)\bigr]u = \int_0^D{ze^{\eta (1-u)ze^{R(T-t)}}\,dF_Z(z)} - \mathbb{E}[Z].
\end{equation}
\end{lemma}
\begin{proof}
From the expression~\eqref{eqn:variance_modified_prop} we get
\[
\frac{\partial q(t,y,u)}{\partial u} = \mathbb{E}[Z]\lambda(t,y) + 2\theta_r\mathbb{E}[Z^2]\bigl[\lambda(t,y)+T\lambda(t,y)^2\bigr]u \qquad \forall u\in(0,1)
\]
and
\[
\frac{\partial^2 q(t,y,u)}{\partial u^2}=2\theta_r\mathbb{E}[Z^2]\bigl[\lambda(t,y)+T\lambda(t,y)^2\bigr]>0 \qquad \forall u\in(0,1).
\]
By Lemma~\ref{PsiUconcave} $\Psi^u(t,y)$ is strictly concave w.r.t. $u$ and full reinsurance is never optimal because of Remark~\ref{remark:convexpremium_fullreins}. Moreover, we notice that $A_0=\emptyset$ as in Lemma~\ref{lemma:variance_vp_prop}, thus the optimal strategy is unique and it belongs to $(0,1)$. In order to find such a solution, we turn the attention to the first order condition, which is exactly equation~\eqref{eqn:first_order_cond_modvariance}.
\end{proof}

Through the numerical simulations in Section~\ref{section:simulation} we will show that the \textit{intensity-adjusted variance premium principle} leads to optimal strategies which are consistent with the desired properties obtained under the other premium calculation principles. Moreover, the reinsurance strategies under the \textit{intensity-adjusted variance premium principle} are not deterministic and explicitly depend on the (stochastic) intensity. Hence the problems described in the beginning of this subsection are fixed. \\

Using the result given in Example~\ref{example:variance_premium_exponential}, it is easy to specialize Lemma~\ref{lemma:variance_modified_prop} to the case of exponentially distributed claims.

\begin{example}{(Exponentially distributed claims under the intensity-adjusted variance premium principle)} \\
Under the \textit{intensity-adjusted variance premium principle}~\eqref{eqn:variance_modified_prop}, suppose that the claims are exponentially distributed with parameter $\zeta>\eta e^{RT}$. Then  the optimal strategy $u^*(t,y)\in(0,1)$ is given by  
\begin{equation}
\label{eqn:exponential_modvarianceprem_opt_control}
u^*(t,y)=1-\frac{\zeta}{\eta}\biggl(1-\sqrt{\frac{\zeta}{\zeta+4\theta_r\bigl[1+T\lambda(t,y)\bigr]}}\biggr)e^{-R(T-t)} \qquad (t,y)\in[0,T]\times\mathbb{R}.
\end{equation}
\end{example}

\begin{remark}
In~\cite{liangetal:commonshock} and~\cite{yuen_liang_zhou:commonshock} the authors used, respectively, the \textit{variance premium} and the \textit{expected value} principles to obtain optimal reinsurance strategies in a risk model with multiple dependent classes of insurance business. In those papers the optimal strategies explicitly depend on the claims intensities, but it is due to the presence of more than one business line, hence our arguments are not valid there.  Nevertheless, in~\cite{yuen_liang_zhou:commonshock} the authors realized that in the diffusion approximation of the classical risk model the \textit{variance premium principle} lead to optimal strategies which do not depend on the claims intensities. In fact, this was the main motivation of their work. Their observation confirms our perplexities of strategies independent on the claims intensity.
\end{remark}


\section{Optimal investment policy}
\label{section:optimal_investment}

\begin{lemma}
The problem
\[
\sup_{w(t,y,p)\in\mathbb{R}}{\Psi^w(t,y,p)}
\]
where $\Psi^w(t,y,p)$ is defined in~\eqref{eqn:Psi_w}, admits a unique solution $w^*(t,y,p)$ for all $(t,y,p)\in[0,T]\times\mathbb{R}\times(0,+\infty)$ given by
\begin{equation}
\label{eqn:opt_inv}
w^*(t,y,p) =  \frac{\mu(t,p)-R}{\eta\sigma(t,p)^2 e^{R(T-t)}} +\frac{p}{\eta e^{R(T-t)}}
\frac{\frac{\partial{\varphi}}{\partial{p}}(t,y,p)}{\varphi(t,y,p)}.
\end{equation}
\end{lemma}
\begin{proof}
Since $\varphi(t,y,p)>0$,  $\Psi^w(t,y,p)$ is strictly concave w.r.t. $w$ and the result follows from the first order condition.
\end{proof}

We emphasize that the optimal $w^*$ is the sum of the classical solution%
\footnote{See e.g. \cite{merton:inv}.}
plus an adjustment term due to the dependence of the risky asset price coefficients on the stochastic process $\{P_t\}$.

\begin{remark}
If $\mu,\sigma$ are continuous function and $\varphi\in\mathcal{C}^{1,2}$, then $w^*$ is a continuous function w.r.t. $(t,y,p)$.
\end{remark}

\begin{corollary}
\label{lemma:optimal_inv_strategy}
Suppose that there exist two functions $f(t,y):[0,T]\times\mathbb{R}\to(0,+\infty)$ and $g(t,p):[0,T]\times(0,+\infty)\to\mathbb{R}$ such that $\varphi(t,y,p) = f(t,y)e^{g(t,p)}$ for all $(t,y,p)\in[0,T]\times\mathbb{R}\times(0,+\infty)$, with $f(t,y)>0$. Then the optimal investment strategy~\eqref{eqn:opt_inv} reads as follows:
\begin{equation}
\label{eqn:opt_inv2}
w^*(t,p) =  \frac{\mu(t,p)-R}{\eta\sigma(t,p)^2 e^{R(T-t)}} + \frac{p}{\eta e^{R(T-t)}}\frac{\partial{g}}{\partial{p}}(t,p).
\end{equation}
\end{corollary}


\section{Verification Theorem}
\label{section:verification_theorem}

Now we conjecture a solution to equation~\eqref{eqn:hjb_prop} of the form $\varphi(t,y,p) = f(t,y)e^{g(t,p)}$, with $f(t,y)>0$.
Using Lemma~\ref{lemma:optimal_inv_strategy}, replacing all the derivatives and performing some calculations, the equation~\eqref{eqn:hjb_prop} reads as follows
\begin{equation}
\label{eqn:final_HJB}
\left\{
\begin{aligned}
&-\frac{\partial{f}}{\partial{t}}(t,y) - b(t,y) \frac{\partial{f}}{\partial{y}}(t,y)- \frac{1}{2}\gamma(t,y)^2\frac{\partial^2{f}}{\partial{y^2}}(t,y) + \biggl[\eta e^{R(T-t)}c(t,y)+\max_{u(t,y)\in[0,1]}{\Psi^u(t,y)}\biggr]f(t,y) \\
&+f(t,y)\biggl[-\frac{\partial{g}}{\partial{t}}(t,p) - pR\frac{\partial{g}}{\partial{p}}(t,p) - \frac{1}{2}p^2\sigma(t,p)^2\frac{\partial^2{g}}{\partial{p^2}}(t,p) + \frac{1}{2}\frac{\bigl(\mu(t,p)-R\bigr)^2}{\sigma(t,p)^2}\biggr] = 0\\
& f(T,y)e^{g(T,p)}=1 \qquad \forall(y,p)\in\mathbb{R}\times(0,+\infty)
\end{aligned}
\right.
\end{equation}
It is easy to show that if $f,g$ are two solutions of the following Cauchy problems
\begin{equation}
\label{eqn:PDEf}
\left\{
\begin{aligned}
&-\frac{\partial{f}}{\partial{t}}(t,y) - b(t,y) \frac{\partial{f}}{\partial{y}}(t,y)- \frac{1}{2}\gamma(t,y)^2\frac{\partial^2{f}}{\partial{y^2}}(t,y) \\
&+ \biggl[\eta e^{R(T-t)}c(t,y)+\max_{u(t,y)\in[0,1]}{\Psi^u(t,y)}\biggr]f(t,y) = 0 \\
&f(T,y)=1
\end{aligned}
\right.
\end{equation}
\begin{equation}
\label{eqn:PDEg}
\left\{
\begin{aligned}
&-\frac{\partial{g}}{\partial{t}}(t,p) - pR\frac{\partial{g}}{\partial{p}}(t,p) - \frac{1}{2}p^2\sigma(t,p)^2\frac{\partial^2{g}}{\partial{p^2}}(t,p) + \frac{1}{2}\frac{\bigl(\mu(t,p)-R\bigr)^2}{\sigma(t,p)^2} = 0\\
&g(T,p)=0
\end{aligned}
\right.
\end{equation}
then they solve the Cauchy problem~\eqref{eqn:final_HJB} and $v(t,x,y,p) = -e^{-\eta xe^{R(T-t)}} f(t,y)e^{g(t,p)}$ solves the original HJB equation given in~\eqref{eqn:HJBformulation_prop}.\\
Before we prove a verification theorem, we must show that our proposed optimal controls are admissible strategies.
\begin{lemma}
\label{lemma:controls_admissible}
Suppose that~\eqref{eqn:PDEf} and~\eqref{eqn:PDEg} admit classical solutions with $\frac{\partial{g}}{\partial{p}}$ satisfying the following growth condition:
\begin{equation}
\label{eqn:dg_growth}
\abs*{\frac{\partial{g}}{\partial{p}}(t,p)} \le C(1+\abs{p}^\beta) \qquad \forall (t,p)\in[0,T]\times(0,+\infty)
\end{equation}
for some constants $\beta>0$ and $C>0$.
Moreover, assume that
\begin{equation}
\label{eqn:P_additional_conditions}
\mathbb{E}\biggl[\int_0^T\abs{\mu(t,P_t)}P_t^{\beta+1}\,dt+\int_0^T \sigma(t,P_t)^2P_t^{2\beta+2}\,dt\biggr]<\infty.
\end{equation}
Let be $u^*(t,y)$ as given in Proposition~\ref{prop:optimalreins} and $w^*(t,p)$ in Lemma~\ref{lemma:optimal_inv_strategy}.
Let us define the processes $u^*_t\doteq u^*(t,Y_t)$ and $w^*_t\doteq w^*(t,P_t)$; then the pair $(u^*_t,w^*_t)$ is an admissible strategy, i.e. $(u^*_t,w^*_t)\in\mathcal{U}$.
\end{lemma}
\begin{proof}
First let us observe that both $u^*_t,w^*_t$ are $\{\mathcal{F}_t\}$-predictable processes since $u^*(t,u)$ and $ w^*(t,p)$ are measurable functions of their arguments and $Y$ is adapted.
Moreover, they take values, respectively, in $[0,1]$ and in $\mathbb{R}$. Furthermore, using the expression~\eqref{eqn:opt_inv2} we have that
\begin{align*}
\mathbb{E}\biggl[\int_0^T{\abs{w^*_t}\abs{\mu(t,P_t)-R}\,dt}\biggr]
&\le \mathbb{E}\biggl[\int_0^T{\frac{(\mu(t,P_t)-R)^2}{\eta\sigma(t,P_t)^2 e^{R(T-t)}}\,dt}\biggr]\\
&+\mathbb{E}\biggl[\int_0^T{\frac{\abs{\mu(t,P_t)-R}P_t}{\eta e^{R(T-t)}}\abs*{\frac{\partial{g}}{\partial{p}}(t,P_t)}\,dt}\biggr] \\
&\le \mathbb{E}\biggl[\int_0^T{\frac{(\mu(t,P_t)-R)^2}{\eta\sigma(t,P_t)^2 e^{R(T-t)}}\,dt}\biggr]\\
&+C\,\mathbb{E}\biggl[\int_0^T{\abs{\mu(t,P_t)}(1+P_t^\beta)P_t\,dt}\biggr]<\infty
\end{align*}
and
\begin{align*}
\mathbb{E}\biggl[\int_0^T{(w^*_t\sigma(t,P_t))^2\,dt}\biggr]
&= \mathbb{E}\biggl[\int_0^T{\frac{(\mu(t,p)-R)^2}{\eta^2 \sigma(t,p)^2 e^{2R(T-t)}}\,dt}\biggr] \\
&+\mathbb{E}\biggl[\int_0^T{\frac{\sigma(t,P_t)^2 P_t^2}{\eta^2 e^{2R(T-t)}}\biggl(\frac{\partial{g}}{\partial{p}}(t,P_t)\biggr)^2\,dt}\biggr] \\
&+ 2\mathbb{E}\biggl[\int_0^T{\frac{(\mu(t,p)-R)P_t}{\eta^2 e^{2R(T-t)}}\frac{\partial{g}}{\partial{p}}(t,P_t)\,dt}\biggr]\\
&\le \mathbb{E}\biggl[\int_0^T{\frac{(\mu(t,p)-R)^2}{\eta^2 \sigma(t,p)^2 e^{2R(T-t)}}\,dt}\biggr] \\
&+C\,\mathbb{E}\biggl[\int_0^T{\frac{\sigma(t,P_t)^2 P_t^2}{\eta^2 e^{2R(T-t)}}\biggl(1+P_t^\beta\biggr)^2\,dt}\biggr] \\
&+ C\,\mathbb{E}\biggl[\int_0^T{\frac{\abs{\mu(t,p)-R}P_t}{\eta^2 e^{2R(T-t)}}(1+P_t^\beta)\,dt}\biggr]<\infty
\end{align*}
where $C$ denotes any positive constant and the expectations are finite because of the Novikov condition~\eqref{eqn:novikov} together with~\eqref{eqn:dg_growth} and~\eqref{eqn:P_additional_conditions}.
\end{proof}

Now we are ready for the verification argument.
\begin{theorem}[Verification Theorem]
\label{theorem:verification}
Suppose that~\eqref{eqn:PDEf} and~\eqref{eqn:PDEg} admit bounded classical solutions, respectively $f\in\mathcal{C}^{1,2}((0,T)\times\mathbb{R}))\cap\mathcal{C}([0,T]\times\mathbb{R}))$ and $g\in\mathcal{C}^{1,2}((0,T)\times(0,+\infty))\cap\mathcal{C}([0,T]\times(0,+\infty))$.\\
Let us assume that the conditions~\eqref{eqn:dg_growth} and~\eqref{eqn:P_additional_conditions} hold and suppose that
\begin{equation}
\label{eqn:df_growth}
\abs*{\frac{\partial{f}}{\partial{y}}(t,y)} \le \tilde{C}(1+\abs{y}^\beta) \qquad \forall (t,y)\in[0,T]\times\mathbb{R}
\end{equation}
for some constants $\beta>0$ and $\tilde{C}>0$.
As an alternative, the conditions~\eqref{eqn:dg_growth},~\eqref{eqn:P_additional_conditions} and~\eqref{eqn:df_growth} may be replaced by the boundedness of $\frac{\partial{g}}{\partial{p}}$ and $\frac{\partial{f}}{\partial{y}}$.\\
Then the function $v:V\to\mathbb{R}$ defined by the following
\begin{equation}
\label{eqn:valuefun_prop}
v(t,x,y,p) = -e^{-\eta xe^{R(T-t)}} f(t,y)e^{g(t,p)}
\end{equation}
is the value function of the reinsurance-investment problem and
\[
\alpha^*(t,Y_t,P_t)=(u^*(t,Y_t),w^*(t,P_t))
\]
with $u^*(t,y)$ given in Proposition~\ref{prop:optimalreins} and $w^*(t,p)$ in~\eqref{eqn:opt_inv2} is an optimal control.
\end{theorem}

\begin{proof}
Let $f(t,y):[0,T]\times\mathbb{R}\to(0,+\infty)$ and $g(t,p):[0,T]\times(0,+\infty)\to\mathbb{R}$ be functions satisfying the assumptions required by Theorem~\ref{theorem:verification} and suppose that they are solutions of the Cauchy problems~\eqref{eqn:PDEf} and~\eqref{eqn:PDEg}. Now consider the function $\varphi(t,y,p) = f(t,y)e^{g(t,p)}$. As already observed, it satisfies equation~\eqref{eqn:hjb_prop}, i.e. it is a solution of the problem 
\begin{equation}
\left\{
\begin{aligned}
&\sup_{(u,w)\in [0,1]\times\mathbb{R}}{\mathcal{H}^\alpha \varphi(t,y,p)} = 0\\
& \varphi(t,y,p) = 1 \qquad \forall(y,p)\in\mathbb{R}\times(0,+\infty).
\end{aligned}
\right.
\end{equation}
Now, taking $v(t,x,y,p) = -e^{-\eta xe^{R(T-t)}}\varphi(t,y,p)$, we have that $v$ is a solution of the Cauchy problem~\eqref{eqn:HJBformulation_prop}. This implies that, for any $(t,x,y,p) \in [0,T]\times\mathbb{R}\times\mathbb{R}\times(0,+\infty)$
\[
\mathcal{L}^\alpha v(s,X_{t,x}^\alpha(s),Y_{t,y}(s),P_{t,p}(s)) \le 0 \quad \forall s\in[t,T]
\]
for all $\alpha\in\mathcal{U}_t$, where $\{Y_{t,y}(s)\}_{s\in[t,T]}$ denotes the solution to equation~\eqref{eqn:stochasticfactor} with initial condition $Y_t=y$ and, similarly, $\{P_{t,p}(s)\}_{s\in[t,T]}$ denotes the solution to equation~\eqref{eqn:priceprocess} with initial condition $P_t=p$.\\
Now, from It\^o's formula we have that
\begin{equation}
\label{eqn:itoformula_veriftheor}
v(T,X_{t,x}^\alpha(T),Y_{t,y}(T),P_{t,p}(T))-v(t,x,y,p) = \int_t^T{\mathcal{L}^\alpha v(s,X_{t,x}^\alpha(s),Y_{t,y}(s),P_{t,p}(s))\,ds}+M_T
\end{equation}
where $\{M_r\}_{r\in[t,T]}$ is the following stochastic process:
\begin{multline}
\label{eqn:mt_prop}
M_r=\int_t^r{w_s\sigma(s,P_s)\frac{\partial{v}}{\partial{x}}(s,X_s^\alpha,Y_s,P_s)\,dW^{(P)}_s}\\
+\int_t^r{P_s\sigma(s,P_s)\frac{\partial{v}}{\partial{p}}(s,X_s^\alpha,Y_s,P_s)\,dW^{(P)}_s}
+\int_t^r{\gamma(s,Y_s)\frac{\partial{v}}{\partial{y}}(s,X_s^\alpha,Y_s,P_s)\,dW^{(Y)}_s} \\
+ \int_0^D\int_t^r{\biggl[v(s,X_s^\alpha-(1-u)z,Y_s,P_s)-v(s,X_s^\alpha,Y_s,P_s)\biggr]\biggl(m(ds,dz)-\lambda(s,Y_s)\,dF_Z(z)\biggr)}.
\end{multline}
Now we prove that $\{M_r\}_{r\in[t,T]}$ is an $\{\mathcal{F}_r\}$-local martingale. Since the jump term is a real martingale because $v$ is bounded, we only need to show that
\begin{align*}
&\mathbb{E}\biggl[\int_t^{T\land\tau_n}{\biggl(w_s\sigma(s,P_s)\frac{\partial{v}}{\partial{x}}(s,X_s^\alpha,Y_s,P_s)\biggr)^2\,ds}\biggr] <\infty \\
&\mathbb{E}\biggl[\int_t^{T\land\tau_n}{\biggl(P_s\sigma(s,P_s)\frac{\partial{v}}{\partial{p}}(s,X_s^\alpha,Y_s,P_s)\biggr)^2\,ds}\biggr] <\infty \\
&\mathbb{E}\biggl[\int_t^{T\land\tau_n}{\biggl(\gamma(s,Y_s)\frac{\partial{v}}{\partial{y}}(s,X_s^\alpha,Y_s,P_s)\biggr)^2\,ds}\biggr] <\infty
\end{align*}
for a suitable non-decreasing sequence of stopping times $\{\tau_n\}_{n=1,\dots}$ such that $\lim_{n\to+\infty}{\tau_n}=+\infty$.
Taking into account the expression~\eqref{eqn:valuefun_prop}, we note that
\begin{align*}
\frac{\partial{v}}{\partial{x}}(t,x,y,p) &= \eta e^{R(T-t)}e^{-\eta xe^{R(T-t)}} f(t,y)e^{g(t,p)} \\
\frac{\partial{v}}{\partial{y}}(t,x,y,p) &= -e^{-\eta xe^{R(T-t)}} e^{g(t,p)}\frac{\partial{f}}{\partial{y}}(t,y)\\
\frac{\partial{v}}{\partial{p}}(t,x,y,p) &= -e^{-\eta xe^{R(T-t)}} f(t,y)e^{g(t,p)}\frac{\partial{g}}{\partial{p}}(t,p).
\end{align*}
Let us define a sequence of random times $\{\tau_n\}_{n=1,\dots}$ as follows:
\[
\tau_n\doteq\inf\{s\in[t,T]\mid X_s^\alpha<-n \lor \abs{Y_s}>n\} \qquad n=1,\dots
\]
In the sequel of the proof we denote with $C_n$ any constant depending on $n=1,\dots$.\\
Then we have that
\begin{align*}
&\mathbb{E}\biggl[\int_0^{T\land\tau_n}{\biggl(w_s\sigma(s,P_s)\frac{\partial{v}}{\partial{x}}(s,X_s^\alpha,Y_s,P_s)\biggr)^2\,ds}\biggr] \\
&=\mathbb{E}\biggl[\int_0^{T\land\tau_n}{\biggl(w_s\sigma(s,P_s)\eta e^{R(T-s)}e^{-\eta X_s^\alpha e^{R(T-s)}} f(s,Y_s)e^{g(s,P_s)}\biggr)^2\,ds}\biggr] \\
&\le C_n\,\mathbb{E}\biggl[\int_0^T{\biggl(w_s\sigma(s,P_s)\biggr)^2\,ds}\biggr] <\infty \qquad \forall n=1,\dots
\end{align*}
because $w_t$ is admissible and $f$ and $g$ are bounded by hypothesis. Moreover,we have that
\begin{align*}
&\mathbb{E}\biggl[\int_0^{T\land\tau_n}{\biggl(\gamma(s,Y_s)\frac{\partial{v}}{\partial{y}}(s,X_s^\alpha,Y_s,P_s)\biggr)^2\,ds}\biggr] \\
&=\mathbb{E}\biggl[\int_0^{T\land\tau_n}{\biggl(\gamma(s,Y_s)e^{-\eta X_s^\alpha e^{R(T-s)}} e^{g(s,P_s)}\frac{\partial{f}}{\partial{y}}(s,Y_s)\biggr)^2\,ds}\biggr] \\
&\le\tilde{C}\mathbb{E}\biggl[\int_0^{T\land\tau_n}{\biggl(\gamma(s,Y_s)e^{-\eta X_s^\alpha e^{R(T-s)}} e^{g(s,P_s)}\biggr)^2(1+\abs{Y_s}^\beta)^2\,ds}\biggr] \\
&\le C_n\,\mathbb{E}\biggl[\int_0^T{\gamma(s,Y_s)^2\,ds}\biggr] <\infty \qquad \forall n=1,\dots
\end{align*}
because $g$ is bounded and using the assumptions~\eqref{eqn:solutionY} and~\eqref{eqn:df_growth}.
Finally, we obtain that
\begin{align*}
&\mathbb{E}\biggl[\int_0^{T\land\tau_n}{\biggl(P_s\sigma(s,P_s)\frac{\partial{v}}{\partial{p}}(s,X_s^\alpha,Y_s,P_s)\biggr)^2\,ds}\biggr] \\
&=\mathbb{E}\biggl[\int_0^{T\land\tau_n}{P_s^2\sigma(s,P_s)^2\biggl(e^{-\eta X_s^\alpha e^{R(T-s)}} f(s,Y_s)e^{g(s,P_s)}\frac{\partial{g}}{\partial{p}}(s,P_s)\biggr)^2\,ds}\biggr]\\
&\le C\mathbb{E}\biggl[\int_0^{T\land\tau_n}{P_s^2\sigma(s,P_s)^2\biggl(e^{-\eta X_s^\alpha e^{R(T-s)}} f(s,Y_s)e^{g(s,P_s)}\biggr)^2(1+\abs{P_s}^\beta)^2\,ds}\biggr]\\
&\le C_n\mathbb{E}\biggl[\int_0^T{\sigma(s,P_s)^2(P_s^2+P_s^{2\beta+2})\,ds}\biggr] <\infty \qquad \forall n=1,\dots
\end{align*}
because $f$ and $g$ are bounded by hypothesis and using conditions~\eqref{eqn:solutionP},~\eqref{eqn:dg_growth} and~\eqref{eqn:P_additional_conditions}.
Thus $\{M_r\}_{r\in[t,T]}$ is an $\{\mathcal{F}_r\}$-local martingale and $\{\tau_n\}_{n=1,\dots}$ is a localizing sequence for $\{M_r\}_{r\in[t,T]}$.\\
Taking the expected value of both sides of~\eqref{eqn:itoformula_veriftheor} with $T$ replaced by $T\land \tau_n$, we obtain that
\[
\mathbb{E}[v(T\land\tau_n,X_{t,x}^\alpha(T\land\tau_n),Y_{t,y}(T\land\tau_n),P_{t,p}(T\land\tau_n))\mid\mathcal{F}_t]
\le v(t,x,y,p) \qquad 
\]
for any $\alpha\in\mathcal{U}_t, t\in[0,T\land\tau_n], n\ge1$.
Now notice that
\begin{align*}
&\mathbb{E}[v(T\land\tau_n,X_{t,x}^\alpha(T\land\tau_n),Y_{t,y}(T\land\tau_n),P_{t,p}(T\land\tau_n))^2]\\
&=\mathbb{E}[e^{-2\eta X_{t,x}^\alpha(T\land\tau_n)e^{R(T\land\tau_n-t)}} f(T\land\tau_n,Y_{T\land\tau_n})^2e^{2g(T\land\tau_n,P_{T\land\tau_n})}]\\
&\le C\,e^{-2\eta ne^{R(T\land\tau_n)}}\le C
\end{align*}
thus $\{v(T\land\tau_n,X_{t,x}^\alpha(T\land\tau_n),Y_{t,y}(T\land\tau_n),P_{t,p}(T\land\tau_n))\}_{n=1,\dots}$ is a family of uniformly integrable random variables. Hence it converges almost surely. Observing that $\{\tau_n\}_{n=1,\dots}$ is a bounded and non-decreasing sequence, since $\mathbb{P}[\abs{X_t^\alpha}<+\infty]=1$ (see~\eqref{eqn:wealth_sol}) and using~\eqref{eqn:solutionY2} and~\eqref{eqn:solutionP2}, taking the limit for $n\to+\infty$, we conclude that
\begin{align}
&\mathbb{E}[v(T,X_{t,x}^\alpha(T),Y_{t,y}(T),P_{t,p}(T))\mid\mathcal{F}_t]\notag\\
&= \lim_{n\to+\infty}{\mathbb{E}[v(T\land\tau_n,X_{t,x}^\alpha(T\land\tau_n),Y_{t,y}(T\land\tau_n),P_{t,p}(T\land\tau_n))\mid\mathcal{F}_t]}\notag\\
&\le v(t,x,y,p) \qquad \forall \alpha\in\mathcal{U}_t, t\in[0,T]. \label{eqn:verification_inequality}
\end{align}
To be precise, we have that
\[
\lim_{n\to+\infty}{X_{t,x}^\alpha(T\land\tau_n)} = X_{t,x}^\alpha(T-) = X_{t,x}^\alpha(T) \qquad \mathbb{P}\text{-a.s.}
\]
since the jump of $\{N_t\}_{t\in[0,T]}$ occurs at time $T$ with probability zero.
Using the final condition of the HJB equation~\eqref{eqn:HJBformulation_prop}, from~\eqref{eqn:verification_inequality} we get
\[
\mathbb{E}[U(X_{t,x}^\alpha(T))]\le v(t,x,y,p) \qquad \forall \alpha\in\mathcal{U}_t, t\in[0,T].
\]
Now note that $\alpha^*(t,y,p)$ was calculated in order to obtain $\mathcal{L}^{\alpha^*} v(t,x,y,p)=0$; replicating the calculations above, replacing $\mathcal{L}^\alpha$ with $\mathcal{L}^{\alpha^*}$, we find the equality:
\[
\sup_{\alpha\in\mathcal{U}_t} \mathbb{E}[U(X_{t,x}^\alpha(T))\mid Y_t=y,P_t=p]=v(t,x,y,p) 
\]
thus $\alpha^*(t,Y_t,P_t)$ is an optimal control.
\end{proof}

After the characterization of the value function, we provide a probabilistic representation by means of the Feynman-Kac formula. In preparation for this result, let us introduce a new probability measure $\mathbb{Q}\ll\mathbb{P}$. Novikov condition~\eqref{eqn:novikov} implies that the process $\{L_t\}_{t\in[0,T]}$ defined by
\[
L_t = e^{-\bigl(\frac{1}{2}\int_0^t \abs{\frac{\mu(s,P_s)-R}{\sigma(s,P_s)}}^2\,ds+\int_0^t \frac{\mu(s,P_s)-R}{\sigma(s,P_s)}\,dW^{(P)}_s\bigr)}
\]
is an $\{\mathcal{F}_t\}$-martingale and we can introduce the following probability measure $\mathbb{Q}$:
\begin{equation}
\label{eqn:measureQ}
\frac{d\mathbb{Q}}{d\mathbb{P}}\bigg|_{\mathcal{F}_t}=L_t \qquad t\in[0,T].
\end{equation}
By Girsanov theorem we know that $\tilde{W}^{(P)}_t=W^{(P)}_t+\int_0^t \frac{\mu(s,P_s)-R}{\sigma(s,P_s)}\,ds$ is a $\mathbb{Q}$-Brownian motion and we can rewrite the risky asset dynamic as
\begin{equation}
\label{eqn:riskyassetQ}
d\tilde{P}_t=P_t\biggl[R\,dt + \sigma(t,P_t)\,d\tilde{W}^{(P)}_t\biggr].
\end{equation}
Since the discounted price $\{\tilde{P}_t=P_t e^{-Rt}\}_{t\in[0,T]}$ turns out to be an $\{\mathcal{F}_t\}$-martingale, then $\mathbb{Q}$ is a martingale or risk-neutral measure for $\{P_t\}$%
\footnote{Let us observe that under $\mathbb{Q}$ the dynamics of $\{Y_t\}$ and $\{R_t\}$ do not change.}.
We will denote by $\mathbb{E}^\mathbb{Q}$ the conditional expectation with respect to the martingale measure $\mathbb{Q}$.

\begin{prop}
Suppose that~\eqref{eqn:PDEf} and~\eqref{eqn:PDEg} admit classical solutions $f\in\mathcal{C}^{1,2}((0,T)\times\mathbb{R}))\cap\mathcal{C}([0,T]\times\mathbb{R}))$ and $g\in\mathcal{C}^{1,2}((0,T)\times(0,+\infty))\cap\mathcal{C}([0,T]\times(0,+\infty))$, respectively, both bounded with $\frac{\partial{f}}{\partial{y}}$ and $\frac{\partial{g}}{\partial{p}}$ satisfying the growth conditions~\eqref{eqn:df_growth} and~\eqref{eqn:dg_growth}. Then $f$ and $g$ admit the following Feynman-Kac representations:
\begin{equation}
\label{eqn:f_feynmankac}
f(t,y) = \mathbb{E}\biggl[e^{-\int_t^T{\bigl(\eta e^{R(T-s)}c(s,Y_s)+\Psi^{u^*}(s,Y_s)\bigr) \,ds}}\mid Y_t=y \biggr]
\end{equation}
\begin{equation}
\label{eqn:g_feynmankac}
g(t,p) = -\mathbb{E}^\mathbb{Q}\biggl[\int_t^T{\frac{1}{2}\frac{\bigl(\mu(s,P_s)-R\bigr)^2}{\sigma(s,P_s)^2}\,ds}\mid P_t=p \biggr]
\end{equation}
where $\Psi^{u^*}(t,y)$ is the function defined by~\eqref{eqn:Psi_u}, replacing $u$ with $u^*(t,y)$, and $\mathbb{Q}$ is the probability measure introduced in~\eqref{eqn:measureQ}.
\end{prop}
\begin{proof}
The result is a simple consequence of the Feynman-Kac theorem.
\end{proof}
In Section~\ref{section:PDE} we will provide sufficient conditions which ensure that the functions $f$ and $g$ given in~\eqref{eqn:f_feynmankac} and~\eqref{eqn:g_feynmankac} are, respectively, $\mathcal{C}^{1,2}((0,T)\times\mathbb{R})$ and $\mathcal{C}^{1,2}((0,T)\times(0,+\infty))$ solutions to the Cauchy problems~\eqref{eqn:PDEf} and~\eqref{eqn:PDEg}.


\section{Simulations and numerical results}
\label{section:simulation}

Here we illustrate some numerical results based on the theoretical framework developed in the previous sections. In particular, we perform sensitivity analysis of the optimal reinsurance-investment strategy in order to study the effect of the model parameters on the insurer's decision.


\subsection{Reinsurance strategy}

First, we compare the optimal reinsurance strategy under the \textit{expected value principle} (see Lemma~\ref{lemma:expected_vp_prop}) and the \textit{intensity-adjusted variance premium principle} (see Lemma~\ref{lemma:variance_modified_prop}). In this subsection the first one will be shortly referred as EVP, while the second one as IAVP. The main difference is that under EVP we loose the dependence on the stochastic factor, while under IAVP we keep this dependence; moreover, IAVP also depends on the second moment of the r.v. $Z$ introduced in Section~\ref{section:formulation}.\\
In what follows we assume that $\{Z_i\}_{i=1,\dots}$ is a sequence of i.i.d. positive random variables Pareto distributed with shape parameter $1.8182$ and scale parameter $0.0545$. The stochastic factor is described by the SDE~\eqref{eqn:stochasticfactor} with constant parameters $b=0.3,\gamma=0.3$ and initial condition $Y_0=1$. For the sake of simplicity, we assume that $\lambda(t,y)=\lambda_0 e^{\frac{1}{2}y}$, that is $\{\lambda_t=\lambda(t,Y_t)\}_{t\in[0,T]}$ solves
\[
d\lambda_t=\lambda_t\,\frac{1}{2}dY_t \qquad \lambda_0=0.1,
\]
which guarantees that the intensity is positive. Finally, we consider the model parameters in Table~\ref{tab:parameters}, using the notation introduced in Section~\ref{section:formulation}.
\begin{table}[htp]
\caption{Simulation parameters}
\label{tab:parameters}
\centering
\begin{tabular}{ll}
\toprule
\textbf{Parameter} & \textbf{Value}\\
\midrule
$T$ & $5$ Y\\
$\eta$ & $0.5$\\
$\theta_r$ & $0.1$\\
$R$ & $5\%$\\
\bottomrule
\end{tabular}
\end{table}

From Figure~\ref{img:reins_eta_sens} we observe that the optimal reinsurance strategy is positively correlated to the risk-aversion parameter; moreover, the strategy under EVP seems to be more sensitive to any variation of the risk-aversion.\\

\begin{figure}[htp]
\centering
\includegraphics[width=\textwidth]{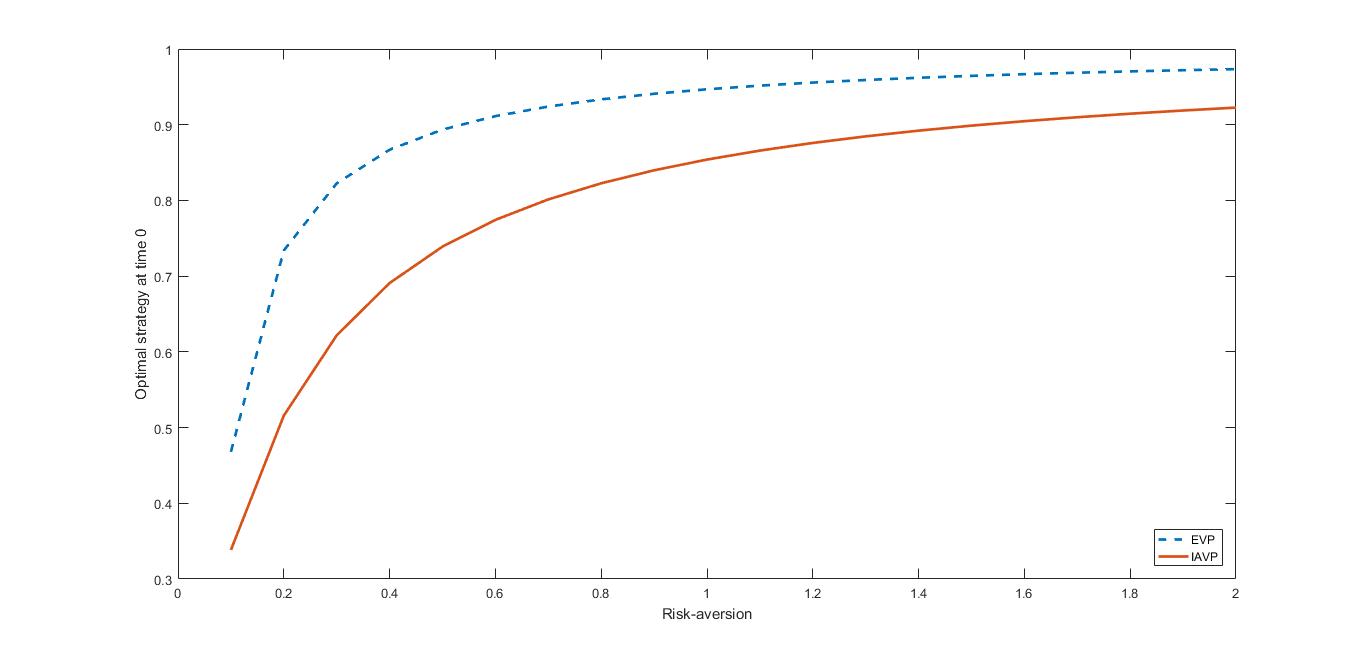}
\caption{The effect of the risk-aversion parameter $\eta$ on the optimal initial strategy}
\label{img:reins_eta_sens}
\end{figure}

In Figure~\ref{img:reins_theta_sens} we notice that any increase in the reinsurance safety loading leads to a decrease of the reinsured risks. It is a simple consequence of the well-known law of demand: the higher the price, the lower the quantity demanded. It is worth noting that under our assumptions the strategy under IAVP is more sensitive than under EVP.

\begin{figure}[htp]
\centering
\includegraphics[width=\textwidth]{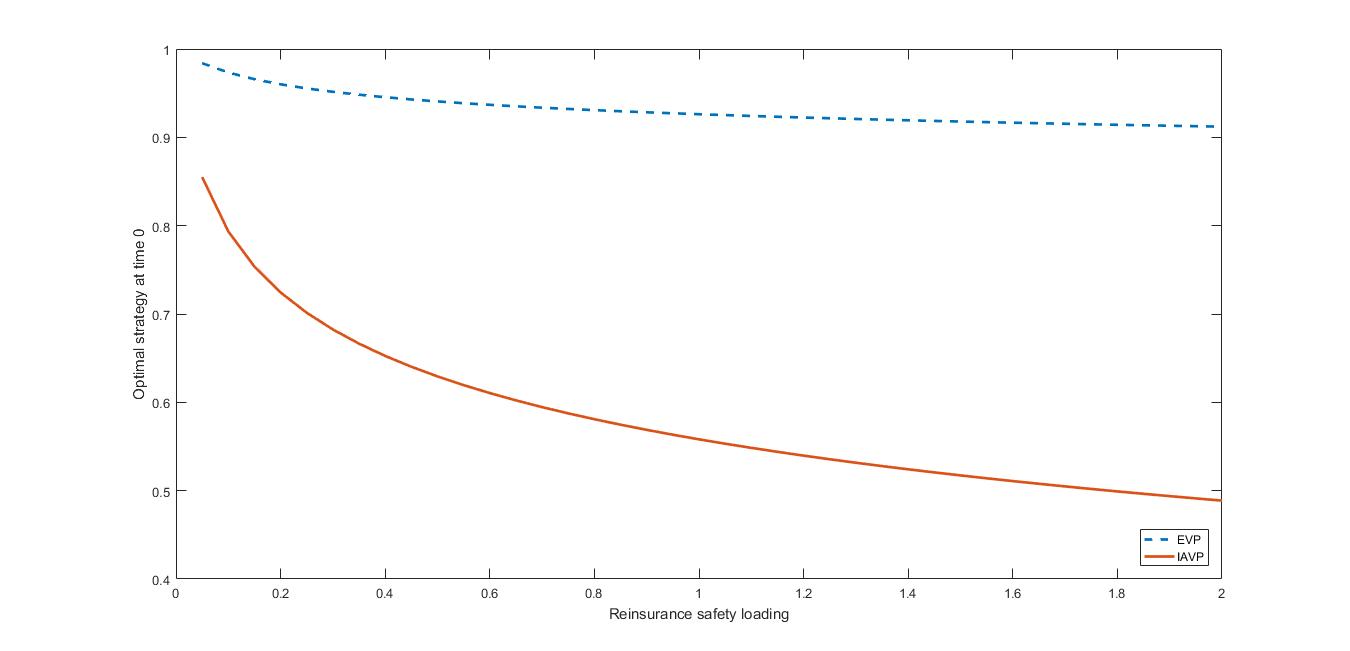}
\caption{The effect of the reinsurance safety loading $\theta_r$ on the optimal initial strategy}
\label{img:reins_theta_sens}
\end{figure}

Finally, in Figure~\ref{img:reins_T_sens} we can see that the insurer increases the protection when the time horizon is higher. Again, the strategy under EVP turns out to be more sensitive to any change of the time horizon. It is interesting that over 15 years EVP leads to more conservative strategies.

\begin{figure}[htp]
\centering
\includegraphics[width=\textwidth]{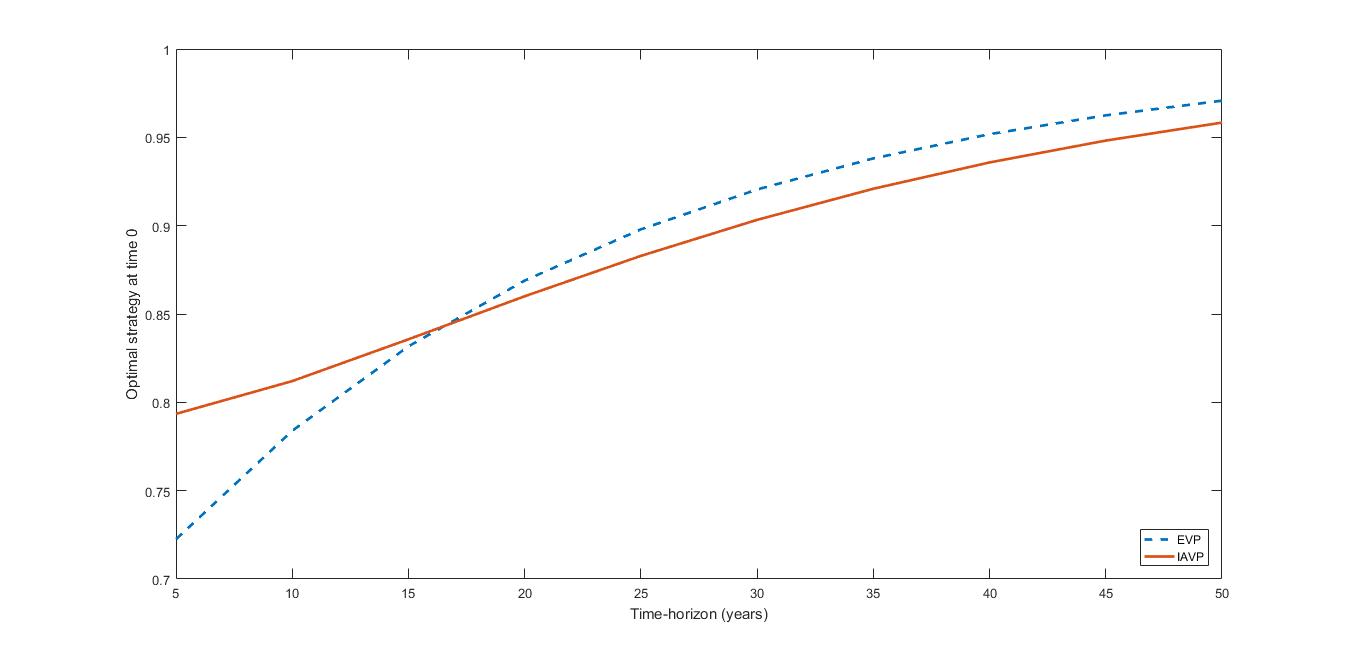}
\caption{The effect of the time horizon $T$ on the optimal initial strategy}
\label{img:reins_T_sens}
\end{figure}

We conclude this subsection investigating the dynamical properties of the reinsurance strategies under EVP and IAVP%
\footnote{Under a practical point of view, we simulated the stochastic processes using the classical Euler's approximation method, with $dt=\frac{T}{500}$.}.
Figure~\ref{img:reins_dynamic} shows that the mean behavior of the optimal reinsurance strategy is decreasing over the time interval; nevertheless, under IAVP the strategy crucially depends on the stochastic factor, hence the insurer will react to any movement of the claims intensity, while under EVP she will follow a deterministic strategy.

\begin{figure}[htp]
\centering
\includegraphics[width=\textwidth]{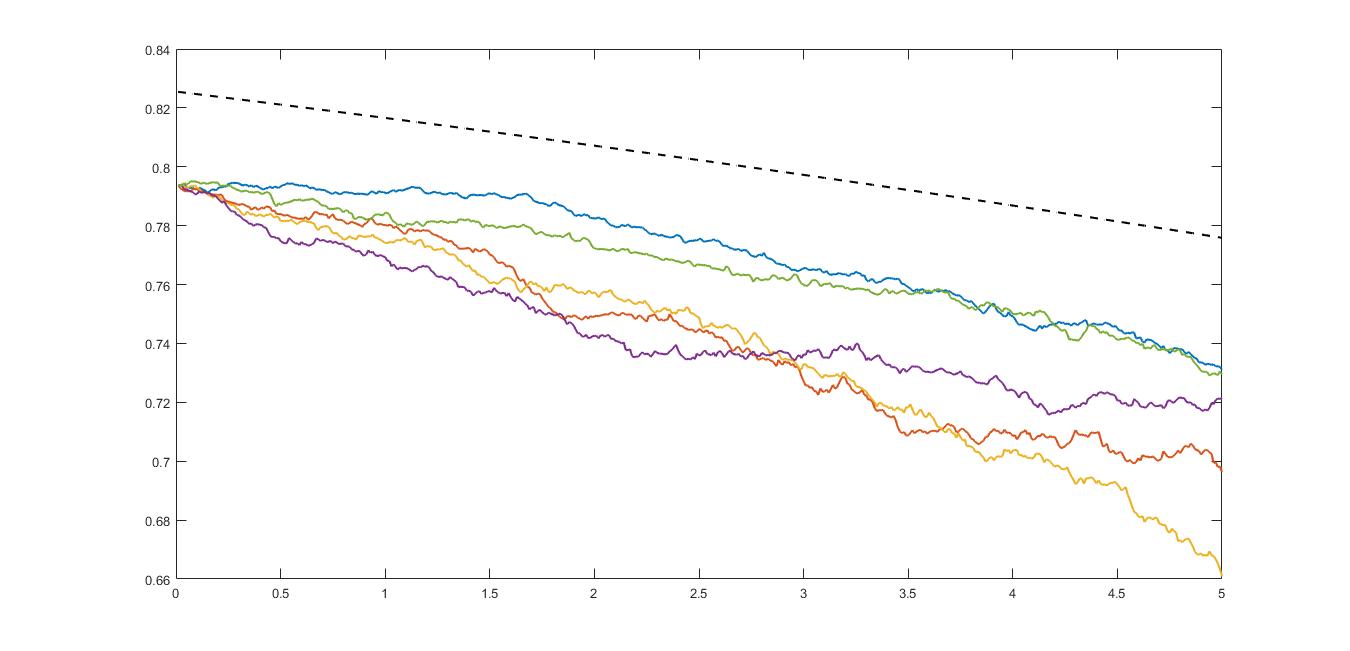}
\caption{Dynamical reinsurance strategies. The dashed line represents the optimal (deterministic) strategy under EVP.}
\label{img:reins_dynamic}
\end{figure}

Summarizing the main results of our numerical simulations, we can conclude that any variation of the model parameters has the same effect on the optimal strategy under EVP and IAVP, at least from a qualitative point of view. It is important noting that any quantitative comparison is affected by the parameters initial choice. Nevertheless, we can state that using our model parameters under EVP the strategy is more sensitive with respect to the model parameters, except for the safety loading, but it is dynamically more stable during the time interval $[0,T]$, because it does not take into account any variation of the claims intensity.


\subsection{Investment strategy}

Now we illustrate a sensitivity analysis for the investment strategy based on the Corollary~\ref{lemma:optimal_inv_strategy}. In our simulations we assumed that the risky asset follows a CEV model, that is
\[
dP_t = P_t \biggl[\mu\,dt + \sigma P_t^\beta\,dW^{(P)}_t\biggr] \qquad P_0=1
\]
with $\mu=0.1,\sigma=0.1,\beta=0.5$, while the risk-free interest rate is $R=5\%$ as in the previous subsection. Let us observe that this model corresponds to~\eqref{eqn:priceprocess} assuming that $\mu(t,p)=\mu$ and $\sigma(t,p)=\sigma p^\beta$, with constant $\mu,\sigma>0$.
The numerical computation of the function $g(t,p)$ and its partial derivative $\frac{\partial{g}}{\partial{p}}(t,p)$ is required by the equation~\eqref{eqn:opt_inv2}; for this purpose we used the Feynman-Kac representation given in~\eqref{eqn:g_feynmankac} evaluated through the standard Monte Carlo method.\\

In figure~\ref{img:eta_inv} we show that the higher is the insurer's risk aversion, the lower is the total amount invested in the risky asset.

\begin{figure}[htp]
\centering
\includegraphics[width=\textwidth]{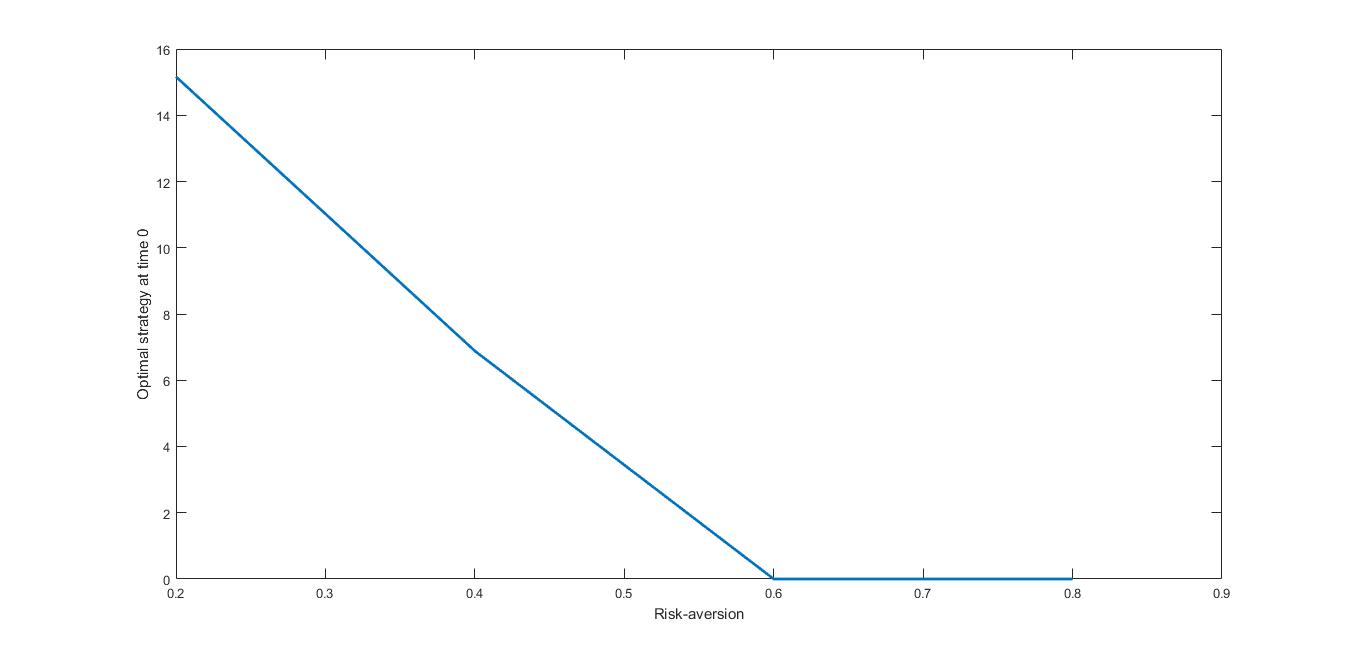}
\caption{The effect of the risk-aversion parameter $\eta$ on the optimal initial strategy}
\label{img:eta_inv}
\end{figure}

Figure~\ref{img:sigma_inv} illustrates that if the volatility increases, then an increasing portion of the insurer's wealth is invested in the risk-free asset.

\begin{figure}[htp]
\centering
\includegraphics[width=\textwidth]{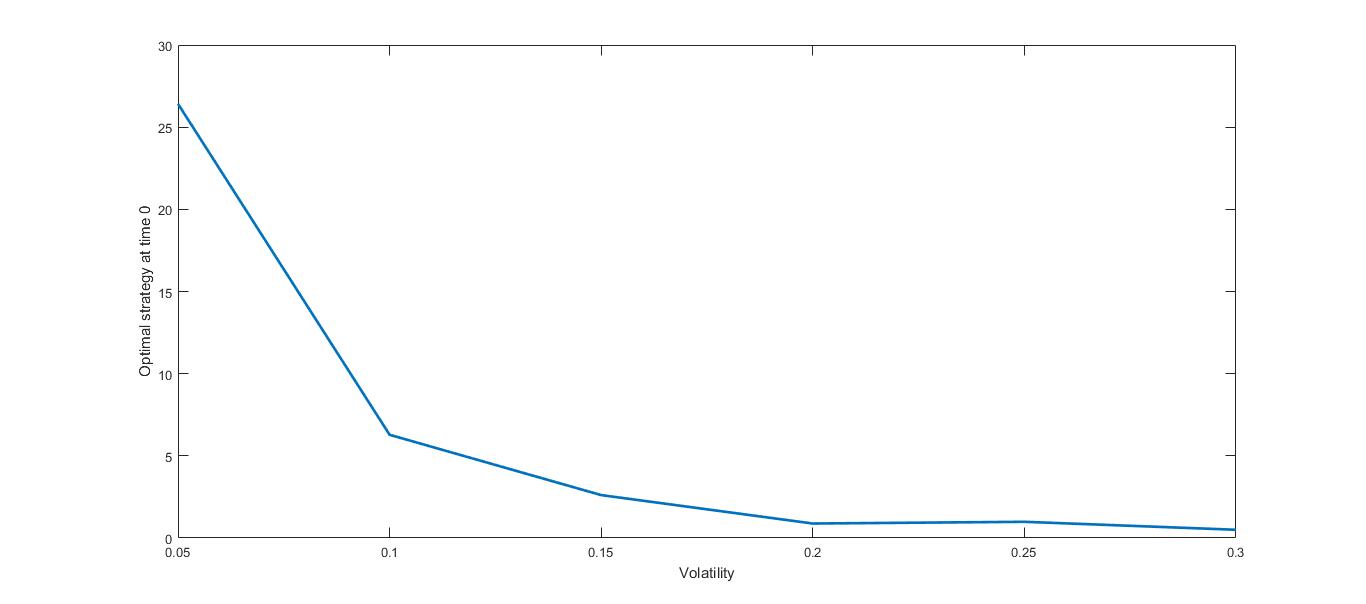}
\caption{The effect of the volatility parameter $\sigma$ on the optimal initial strategy}
\label{img:sigma_inv}
\end{figure}

Finally, if the risk-free interest rate grows up, then the insurer will find it more convenient to invest its surplus in the risk-free asset, as shown in figure~\ref{img:R_inv}.

\begin{figure}[htp]
\centering
\includegraphics[width=\textwidth]{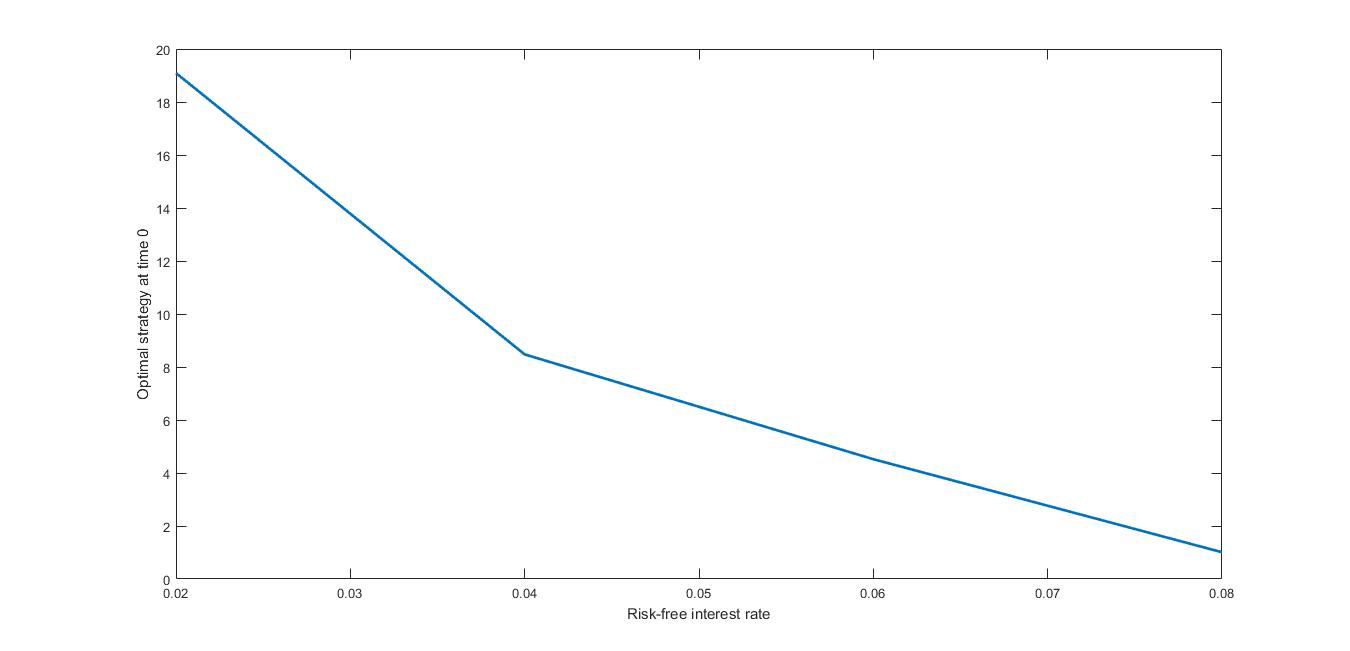}
\caption{The effect of the risk-free interest rate $R$ on the optimal initial strategy}
\label{img:R_inv}
\end{figure}

Similar results can be found in~\cite{shengrongzhao:optreins}. In particular, figure~\ref{img:sigma_inv} confirms the result obtained in Figure 3a of that paper; in addition, figures~\ref{img:eta_inv} and~\ref{img:R_inv} completes the sensitivity analyses performed there.

\pagebreak
\section{Existence and uniqueness of classical solutions}
\label{section:PDE}

In this section we are interested in providing sufficient conditions for existence and uniqueness of the solutions to the PDEs involved in the reinsurance-investment problem, see the Cauchy problems~\eqref{eqn:PDEf} and~\eqref{eqn:PDEg} and as a consequence of a classical solution to HJB equation associated with our problem.\\
First, let us consider~\eqref{eqn:PDEg}. The following Lemma prepares the main result.
\begin{lemma}
\label{lemma:k_holder_forg}
Let us define the set $D_n\doteq(\frac{1}{n},n)$ for $n=1,\dots$ and assume that the functions $\mu(t,p),\sigma(t,p)$ are Lipschitz-continuous in $p\in\overline{D_n}$, uniformly in $t\in[0,T]$. Moreover, assume that $\sigma(t,p)$ is bounded from below, i.e. there exists a constant $\delta_\sigma>0$ such that $\sigma(t,p)\ge\delta_\sigma$ for all $(t,p)\in[0,T]\times(0,+\infty)$.\\
Then for each $n=1,\dots$ the function $k:[0,T]\times(0,+\infty)\to\mathbb{R}$ defined by
\begin{equation}
\label{eqn:function_k_forg}
k(t,p) = \frac{\bigl(\mu(t,p)-R\bigr)^2}{\sigma(t,p)^2}
\end{equation}
is uniformly Lipschitz-continuous on $[0,T]\times\overline{D_n}$.
\end{lemma}
\begin{proof}
Firstly, using the Lipschitz-continuity of the parabolic function on the bounded domain $\overline{D_n}$ we have that
\begin{align*}
\abs*{k(t,p)-k(t',p')} &= \abs*{\biggl(\frac{\mu(t,p)-R}{\sigma(t,p)}\biggr)^2-\biggl(\frac{\mu(t',p')-R}{\sigma(t',p')}\biggr)^2}\\
&\le K_n\abs*{\frac{\mu(t,p)-R}{\sigma(t,p)}-\frac{\mu(t',p')-R}{\sigma(t',p')}}\\
&= K_n\abs*{\frac{\sigma(t',p')[\mu(t,p)-R]-\sigma(t,p)[\mu(t',p')-R]}{\sigma(t,p)\sigma(t',p')}}
\end{align*}
for a positive constant $K_n>0$ which depends on $n$. Now, being $\sigma(t,p)$ bounded from below, setting $\tilde{K_n}=\frac{K_n}{\delta_\sigma^2}$ we have that
\begin{align*}
\abs*{k(t,p)-k(t',p')} &\le\tilde{K_n}\abs*{\sigma(t',p')[\mu(t,p)-R]-\sigma(t,p)[\mu(t',p')-R]}\\
&\le\tilde{K_n}R\abs*{\sigma(t,p)-\sigma(t',p')}+\tilde{K_n}\abs*{\sigma(t',p')\mu(t,p)-\sigma(t,p)\mu(t',p')}\\
&\le\tilde{K_n}R\abs*{\sigma(t,p)-\sigma(t',p')}+\tilde{K_n}\abs*{\sigma(t',p')\mu(t,p)-\sigma(t',p')\mu(t',p')}\\
&\,+\tilde{K_n}\mu(t',p')\abs*{\sigma(t',p')-\sigma(t,p)}
\end{align*}
and, observing that any Lipschitz-continuous function on a bounded domain is also bounded, the result is a consequence of our hypotheses.
\end{proof}

\begin{theorem}
\label{theorem:g_pde}
Suppose that the following conditions are satisfied:
\begin{enumerate}
\item $\mu(t,p)$ and $\sigma(t,p)$ are locally Lipschitz-continuous in $p$, uniformly in $t\in[0,T]$, i.e. for each $n=1,\dots$ there exists a positive constant $K_n$ such that
\[
\abs{\mu(t,p)-\mu(t,p')}+\abs{\sigma(t,p)-\sigma(t,p')}\le K_n\abs{p-p'} \qquad \forall p,p'\in \biggl[\frac{1}{n},n\biggr],t\in[0,T];
\]
\item for all couple $(t,p)\in[0,T]\times(0,+\infty)$ the solution $\{P_{t,p}(s)\}_{s\in[t,T]}$ does not explode, i.e.
\[
\mathbb{P}[\sup_{s\in[t,T]}{P_{t,p}(s)}<\infty]=1;
\]
for instance, it is true if we assume the sub-linear growth for $\sigma(t,p)$:
\[
\abs{\sigma(t,p)} \le K_\sigma(1+p) \qquad \forall p\in(0,+\infty),t\in[0,T]
\]
together with the other hypotheses of this theorem;%
\footnote{See~\cite{pascucci:PDEMartingale}, Theorem 9.11, p. 281.}
\item $\sigma(t,p)$ is bounded from below, i.e. there exists a constant $\delta_\sigma>0$ such that $\sigma(t,p)\ge\delta_\sigma$ for all $(t,p)\in[0,T]\times(0,+\infty)$;
\item $\mu(t,p)$ is bounded, i.e. there exists a constant $\delta_\mu>0$ such that $\abs{\mu(t,p)}\le\delta_\mu$ for all $(t,p)\in[0,T]\times(0,+\infty)$.
\end{enumerate}
Then the function $g(t,p)$ given in~\eqref{eqn:g_feynmankac} satisfies the Cauchy problem~\eqref{eqn:PDEg} and there exists a unique classical solution to~\eqref{eqn:PDEg}. Moreover, we have that $g\in\mathcal{C}^{1,2}((0,T)\times(0,+\infty))$.
\end{theorem}
\begin{proof}
The proof is a consequence of Theorem 1 and Lemma 2 in~\cite{heath:pde}, toghether with Lemma~\ref{lemma:k_holder_forg}. We highlight that in order to use those results, we take $D_n\doteq(\frac{1}{n},n)$ for $n=1,\dots$ as bounded domains such that $(0,+\infty)=\bigcup_{n=1}^\infty{D_n}$. Moreover, we observe that the function $k$ defined in~\eqref{eqn:function_k_forg} is bounded, as requested in Lemma 2 of~\cite{heath:pde}, because $\mu(t,p)$ is bounded and $\sigma(t,p)$ is bounded from below. 
\end{proof}

\begin{remark}
In~\cite{shengrongzhao:optreins} the authors found an explicit solution of the Cauchy problem~\eqref{eqn:PDEg} in the particular case of the CEV model, i.e. when $\mu(t,p)=\mu$ and $\sigma(t,p)=kp^\beta$.
\end{remark}


Now we turn the attention to the second PDE involved in the reinsurance-investment problem, see the Cauchy problem~\eqref{eqn:PDEf}.
Before proving the existence theorem, let us state some preliminary results.
\begin{lemma}
\label{lemma:H(t,y,u)_holder}
Given a compact set $K\in\mathbb{R}$ let us assume that $H(t,y,u):[0,T]\times\mathbb{R}\times K$ is H\"{o}lder-continuous in $(t,y)\in[0,T]\times\mathbb{R}$ uniformly in $u\in K$ with exponent $0<\xi\le1$. Then $\max_{u\in K}{H(t,y,u)}$ is H\"{o}lder-continuous in $(t,y)\in[0,T]\times\mathbb{R}$ with exponent $0<\xi\le1$.
\end{lemma}
\begin{proof}
Given $t,t'\in[0,T]$ and $y,y'\in\mathbb{R}$, let us define
\[
h_1(u)=H(t,y,u) \qquad h_2(u)=H(t',y',u).
\]
Then we have that
\begin{equation}
\label{eqn:max_holder}
\abs{\max_{u\in K}{h_1(u)}-\max_{u\in K}{h_2(u)}}\le \max_{u\in K}{\abs{h_1(u)-h_2(u)}}.
\end{equation}
In fact, observing that
\[ \abs{\max_{u\in K}{h_1(u)}-\max_{u\in K}{h_2(u)}} =
\begin{cases}
      \max_{u\in K}{h_1(u)}-\max_{u\in K}{h_2(u)} & \text{if} \quad \max_{u\in K}{h_1(u)}\ge\max_{u\in K}{h_2(u)} \\
      \max_{u\in K}{h_2(u)}-\max_{u\in K}{h_1(u)} & \text{if} \quad \max_{u\in K}{h_1(u)}<\max_{u\in K}{h_2(u)} \\
   \end{cases}
\]
we notice that in the first case
\begin{align*}
\max_{u\in K}{h_1(u)}-\max_{u\in K}{h_2(u)} &= \max_{u\in K}{[h_1(u)-h_2(u)+h_2(u)]}-\max_{u\in K}{h_2(u)}\\
&\le \max_{u\in K}{[h_1(u)-h_2(u)]}\\
&\le \max_{u\in K}{\abs{h_1(u)-h_2(u)}},
\end{align*}
and in the second case we have that
\begin{align*}
\max_{u\in K}{h_2(u)}-\max_{u\in K}{h_1(u)}&\le \max_{u\in K}{[h_2(u)-h_1(u)]}\\
&\le \max_{u\in K}{\abs{h_1(u)-h_2(u)}}.
\end{align*}
Now, using inequality~\eqref{eqn:max_holder}, we have that
\begin{align*}
\abs{\max_{u\in K}{H(t,y,u)}-\max_{u\in K}{H(t',y',u)}} &\le \max_{u\in K}{\abs{H(t,y,u)-H(t',y',u)}}\\
&\le L(\abs{t-t'}^\xi+\abs{y-y'}^\xi)
\end{align*}
and this completes the proof.
\end{proof}

\begin{corollary}
\label{corollary:Psi_holder}
Let us assume that the following hypotheses hold:
\begin{itemize}
\item $q(t,y,u)$ is bounded and H\"{o}lder-continuous in $(t,y)\in[0,T]\times\mathbb{R}$ uniformly in $u\in[0,1]$ with exponent $0<\xi\le1$;
\item $\lambda(t,y)$ is bounded and H\"{o}lder-continuous in $(t,y)\in[0,T]\times\mathbb{R}$ with exponent $0<\xi\le1$.
\end{itemize}
Then $\max_{u(t,y)\in[0,1]}{\Psi^u(t,y)}$ is H\"{o}lder-continuous in $(t,y)\in[0,T]\times\mathbb{R}$ with exponent $0<\xi\le1$.
\end{corollary}
\begin{proof}
In view of Lemma~\ref{lemma:H(t,y,u)_holder}, it is sufficient to show that $\Psi^u(t,y)$ is H\"{o}lder-continuous in $(t,y)\in[0,T]\times\mathbb{R}$ uniformly in $u\in[0,1]$ with exponent $0<\xi\le1$. Let us recall equation~\eqref{eqn:Psi_u}:
\[
\Psi^u(t,y)= - \eta e^{R(T-t)}q(t,y,u) + \lambda(t,y)\int_0^D{\biggl[1-e^{\eta (1-u)ze^{R(T-t)}}\biggr]\,dF_Z(z)}
\]
Since $e^{R(T-t)}$ is differentiable and bounded on $t\in[0,T]$, our first hypothesis ensures that the first term $\eta e^{R(T-t)}q(t,y,u)$ is H\"{o}lder-continuous in $(t,y)\in[0,T]\times\mathbb{R}$ uniformly in $u\in[0,1]$ with exponent $0<\xi\le1$. For the second term we notice that it is a product of two bounded and H\"{o}lder-continuous functions, in fact
\begin{align*}
&\abs*{\int_0^D{e^{\eta (1-u)ze^{R(T-t)}}\,dF_Z(z)}-\int_0^D{e^{\eta (1-u)ze^{R(T-t')}}\,dF_Z(z)}}\\
&\le \mathbb{E}\biggl[\abs*{e^{\eta (1-u)Ze^{R(T-t)}}-e^{\eta (1-u)Ze^{R(T-t')}}}\biggr].
\end{align*}
Using Lagrange's theorem, there exists $\bar{t}\in[0,T]$ such that
\begin{align*}
&\mathbb{E}\biggl[\abs*{e^{\eta (1-u)Ze^{R(T-t)}}-e^{\eta (1-u)Ze^{R(T-t')}}}\biggr]\\
&\le \mathbb{E}\biggl[\abs*{R\eta (1-u)Ze^{R(T-\bar{t})}e^{\eta (1-u)Ze^{R(T-\bar{t})}}}\biggr]\abs{t-t'}\\
&\le R\eta e^{RT}\mathbb{E}\biggl[Ze^{\eta Ze^{RT}}\biggr]\abs{t-t'}
\end{align*}
and the proof is complete.
\end{proof}

The following theorem is based on the main result of~\cite{heath:pde}.

\begin{theorem}
\label{theorem:f_pde}
Suppose that the following conditions are satisfied:
\begin{enumerate}
\item $b(t,y)$ and $\gamma(t,y)$ are locally Lipschitz-continuous in $y$, uniformly in $t\in[0,T]$, i.e. for each $n=1,\dots$ there exists a positive constant $K_n$ such that
\[
\abs{b(t,y)-b(t,y')}+\abs{\gamma(t,y)-\gamma(t,y')}\le K_n\abs{y-y'} \qquad \forall y,y'\in[-n,n],t\in[0,T];
\]
\item for all couple $(t,y)\in[0,T]\times\mathbb{R}$ the solution $\{Y_{t,y}(s)\}_{s\in[t,T]}$ does not explode, i.e.
\[
\mathbb{P}[\sup_{s\in[t,T]}{Y_{t,y}(s)}<\infty]=1;
\]
for instance, it is true when we assume that for some positive constant $K_2$
\[
\abs{b(t,y)}+\abs{\gamma(t,y)} \le K_2(1+\abs{y}) \qquad \forall t\in[0,T],y\in\mathbb{R}
\]
together with the previous assumption;%
\footnote{See~\cite{pascucci:PDEMartingale}, Theorem 9.11, p. 281.}
\item there exists a constant $\delta_\gamma>0$ such that $\gamma(t,y)^2\ge\delta_\gamma$;
\item the intensity function $\lambda(t,y)$ is bounded and H\"{o}lder-continuous in $(t,y)\in[0,T]\times\mathbb{R}$ with exponent $0<\xi\le1$;
\item the reinsurance premium $q(t,y,u)$ is bounded and H\"{o}lder-continuous in $(t,y)\in[0,T]\times\mathbb{R}$ uniformly in $u\in[0,1]$ with exponent $0<\xi\le1$.
\end{enumerate}
Then the function $f(t,y)$ defined in~\eqref{eqn:f_feynmankac} satisfies the Cauchy problem~\eqref{eqn:PDEf} and there exists a unique classical solution to~\eqref{eqn:PDEf}. Moreover, we have that $f\in\mathcal{C}^{1,2}((0,T)\times\mathbb{R})$.
\end{theorem}
\begin{proof}
The proof is a consequence of Theorem 1 and Lemma 2 in~\cite{heath:pde} together with Corollary~\ref{corollary:Psi_holder}, observing that under our assumptions $\max_{u(t,y)\in[0,1]}{\Psi^u(t,y)}$ is continuous and bounded from above.
\end{proof}


\appendix

\section{Appendix}
\label{section:appendix}

\begin{proof}[Proof of Lemma~\ref{prop:random_measure}]
First, let us start considering all the $[0,D]$-indexed processes $\{H(t,z)\}_{t\in[0,T]}$ of this type:
\[
H(t,z) = \tilde{H}_t\mathbbm{1}_A(z) \qquad t\in[0,T],A\in[0,D]
\]
where $\{\tilde{H}_t\}_{t\in[0,T]}$ is a nonnegative and $\{\mathcal{F}_t\}$-predictable process. Using the independence between $\{N_t\}_{t\in[0,T]}$ and $\{Z_n\}_{n\ge1}$ we have that
\begin{align*}
\mathbb{E}\biggl[\int_0^T\int_0^DH(t,z)\,m(dt,dz)\biggr] &= \mathbb{E}\biggl[\sum_{n\ge1}{\tilde{H}_{T_n}\mathbbm{1}_A(Z_n)\mathbbm{1}_{\{T_n\le T\}}}\biggr]\\
&= \sum_{n\ge1}{\mathbb{P}[Z_n\in A]\mathbb{E}\biggl[\tilde{H}_{T_n}\mathbbm{1}_{\{T_n\le T\}}\biggr]}\\
&= \mathbb{P}[Z\in A]\mathbb{E}\biggl[\sum_{n\ge1}{\tilde{H}_{T_n}\mathbbm{1}_{\{T_n\le T\}}}\biggr]\\
&= \mathbb{P}[Z\in A]\mathbb{E}\biggl[\int_0^T{\tilde{H}_t\lambda_t\,dt}\biggr]\\
&= \mathbb{E}\biggl[\int_0^D\int_0^T{H(t,z)\,dF_Z(z)\lambda_t\,dt}\biggr]
\end{align*}
Using~\cite[App. A1, T4 Theorem, p.263]{bremaud:pointproc} this result can be extended to all nonnegative, $\{\mathcal{F}_t\}$-predictable and $[0,D]$-indexed process $\{H(t,z)\}_{t\in[0,T]}$ and this completes the proof.
\end{proof}

\begin{proof}[Proof of Proposition~\ref{nullcontrol_admissible}]
For any constant strategy $\alpha_t=(u,w)$ with $u\in[0,1]$ and $w\in\mathbb{R}$ we have that
\begin{align}
\label{eqn:valuefun_constantstrategy}
&\mathbb{E}[e^{-\eta X_{t,x}^\alpha(T)}\mid\mathcal{F}_t]=\notag\\
&= e^{-\eta xe^{R(T-t)}}\mathbb{E}[e^{-\eta \int_t^T e^{R(T-s)}[c(s,Y_s)-q(s,Y_s,u)]\,ds}e^{\eta \int_t^T\int_0^D e^{R(T-r)}(1-u)z\,m(dr,dz)}\mid\mathcal{F}_t]\times\notag\\
&\times\mathbb{E}[e^{-\eta \int_t^T e^{R(T-s)}w[\mu(s,P_s)-R]\,ds}e^{-\eta \int_t^T e^{R(T-s)}w\sigma(s,P_s)\,dW^{(P)}_s}\mid\mathcal{F}_t]
\end{align}
because of the independence between the financial and the insurance markets. In particular, for the null strategy $\alpha_t=(0,0)$, using the inequality~\eqref{eqn:premia_bounded_prop}, we have that
\begin{align*}
&\mathbb{E}[e^{-\eta X_{t,x}^{(0,0)}(T)}\mid\mathcal{F}_t]\le\\
&\le e^{-\eta xe^{R(T-t)}}e^{\eta \frac{K}{R}(e^{R(T-t)}-1)}\mathbb{E}[e^{\eta \int_t^T\int_0^D e^{R(T-r)}z\,m(dr,dz)}\mid\mathcal{F}_t].
\end{align*}
Now let us notice that
\begin{align*}
\mathbb{E}[e^{\eta e^{RT}\int_t^T\int_0^D z\,m(dr,dz)}\mid\mathcal{F}_t] &= \mathbb{E}[e^{\eta e^{RT}\sum_{i=N_t}^{N_T}{Z_i}}\mid\mathcal{F}_t]\\
&= \sum_{n\ge N_t}{\mathbb{E}[e^{\eta e^{RT}\sum_{i=N_t}^{n}{Z_i}}\mid\mathcal{F}_t]\,\mathbb{P}[N_T=n\mid\mathcal{F}_t]}\\
&= \sum_{n\ge N_t}{\mathbb{E}\biggl[\prod_{i=N_t}^{n}{e^{\eta e^{RT}Z_i}}\mid\mathcal{F}_t\biggr]\,\mathbb{P}[N_T=n\mid\mathcal{F}_t]}\\
&= \sum_{n\ge N_t}{\mathbb{E}[e^{\eta e^{RT}Z}\mid\mathcal{F}_t]^{(n-N_t)}\,\mathbb{P}[N_T=n\mid\mathcal{F}_t]}\\
&= \sum_{n\ge 0}{\mathbb{E}[e^{\eta e^{RT}Z}]^n\,\mathbb{P}[N_T-N_t=n\mid\mathcal{F}_t]}\\
&= \sum_{n\ge 0}{\mathbb{E}[e^{\eta e^{RT}Z}]^n\,\mathbb{E}\biggl[\frac{\bigl(\int_t^T \lambda_s\,ds\bigr)^n}{n!}e^{-\int_t^T \lambda_s\,ds}\mid\mathcal{F}_t\biggr]}\\
&= \mathbb{E}\bigl[e^{(\mathbb{E}[e^{\eta e^{RT}Z}]-1)\int_t^T \lambda_s\,ds}\mid\mathcal{F}_t\bigr]<\infty
\qquad\Braket{\mathbb{P}=1}
\end{align*}
because of the Assumption~\ref{assumption:EZ_finite}.
\end{proof}

\begin{proof}[Proof of Lemma~\ref{lemma:constant_strategies_prop}]
Assume that there exists a positive constant $K'$ such that
\[
\abs{\mu(t,p)}+\sigma(t,p)\le K' \qquad\forall(t,p)\in[0,T]\times(0,+\infty).
\]
From the proof of the Proposition~\ref{nullcontrol_admissible} (see above), we know that for any constant strategy $\alpha_t=(u,w)$ with $u\in[0,1]$ and $w\in\mathbb{R}$ the equation~\eqref{eqn:valuefun_constantstrategy} holds. Now, using the inequality~\eqref{eqn:premia_bounded_prop}, we have that
\begin{align*}
&\mathbb{E}[e^{-\eta X_{t,x}^\alpha(T)}\mid\mathcal{F}_t]\le\\
&\le e^{-\eta xe^{R(T-t)}}e^{\eta \frac{K}{R}(e^{R(T-t)}-1)}\mathbb{E}[e^{\eta \int_t^T\int_0^D e^{R(T-r)}(1-u)z\,m(dr,dz)}\mid\mathcal{F}_t]\times\\
&\times \mathbb{E}[e^{-\eta \int_t^T e^{R(T-s)}w[\mu(s,P_s)-R]\,ds}e^{-\eta \int_t^T e^{R(T-s)}w\sigma(s,P_s)\,dW^{(P)}_s}\mid\mathcal{F}_t]\\
&\le C\,e^{-\eta x}e^{\eta \frac{K}{R}(e^{R(T-t)}-1)}\mathbb{E}[e^{\eta e^{RT}\int_t^T\int_0^D z\,m(dr,dz)}\mid\mathcal{F}_t]\,
\mathbb{E}[e^{-\eta we^{RT} \int_t^T \sigma(s,P_s)\,dW^{(P)}_s}\mid\mathcal{F}_t]
\end{align*}
where $C$ is a positive constant and the first expectation is finite because of the proof of the Proposition~\ref{nullcontrol_admissible}.\\
Now let us define the stochastic process $\{h_t\}_{t\in[0,T]}$ as
\[
h_t = \eta w e^{RT}\sigma(t,P_t)
\]
and set
\[
L_t = e^{-\int_0^t h_s\,dW^{(P)}_s-\frac{1}{2}\int_0^t h_s^2\,ds}
\]
Since $h_t$ is bounded, the Novikov condition is satisfied:
\[
\mathbb{E}[e^{\frac{1}{2}\int_0^T h_s^2\,ds}]<\infty.
\]
This allows us to introduce a new probability measure $\hat{\mathbb{P}}$ using the change of measure given by
\[
L_t = \frac{d\hat{\mathbb{P}}}{d\mathbb{P}}\bigg|_{\mathcal{F}_t}.
\]
Using the Kallianpur-Striebel formula, we obtain that
\begin{align*}
\mathbb{E}[e^{-\eta we^{RT} \int_t^T \sigma(s,P_s)\,dW^{(P)}_s}\mid\mathcal{F}_t]
&=\mathbb{E}[e^{-\int_t^T h_s\,dW^{(P)}_s}\mid\mathcal{F}_t]\\
&=\frac{\mathbb{E}[L_Te^{\frac{1}{2}\int_t^T h_s^2\,ds}\mid\mathcal{F}_t]}{L_t}\\
&\le \mathbb{E}^{\hat{\mathbb{P}}}[e^{\frac{1}{2}\int_t^T h_s^2\,ds}\mid\mathcal{F}_t]<\infty
\end{align*}
and the proof is complete.
\end{proof}

\begin{proof}[Proof of Lemma~\ref{lemmagenerator}]
Looking at~\eqref{eqn:wealth_proc}, we apply It\^o's formula to the stochastic process $f(t,X_t^\alpha,Y_t,P_t)$:
\[
f(t,X_t^\alpha,Y_t,P_t) = f(0,X_0^\alpha,Y_0,P_0) + \int_0^t{\mathcal{L}^\alpha f(s,X_s^\alpha,Y_s,P_s)\,ds}+m_t
\]
where
\begin{multline}
\label{eqn:mt_prop}
m_t=\int_0^t{w_s\sigma(s,P_s)\frac{\partial{f}}{\partial{x}}(s,X_s^\alpha,Y_s,P_s)\,dW^{(P)}_s}\\
+\int_0^t{P_s\sigma(s,P_s)\frac{\partial{f}}{\partial{p}}(s,X_s^\alpha,Y_s,P_s)\,dW^{(P)}_s}
+\int_0^t{\gamma(s,Y_s)\frac{\partial{f}}{\partial{y}}(s,X_s^\alpha,Y_s,P_s)\,dW^{(Y)}_s} \\
+ \int_0^D\int_0^t{\biggl[f(s,X_s^\alpha-(1-u)z,Y_s,P_s)-f(s,X_s^\alpha,Y_s,P_s)\biggr]\biggl(m(ds,dz)-\lambda(s,Y_s)\,dF_Z(z)\biggr)}.
\end{multline}
We only need to prove that this is an $\{\mathcal{F}_t\}$-martingale.\\
Let us observe that
\begin{align*}
&\mathbb{E}\biggl[\int_0^T{\biggl(w_s\sigma(s,P_s)\frac{\partial{f}}{\partial{x}}(s,X_s^\alpha,Y_s,P_s)\biggr)^2\,ds}\biggr] <\infty \\
&\mathbb{E}\biggl[\int_0^T{\biggl(P_s\sigma(s,P_s)\frac{\partial{f}}{\partial{p}}(s,X_s^\alpha,Y_s,P_s)\biggr)^2\,ds}\biggr] <\infty \\
&\mathbb{E}\biggl[\int_0^T{\biggl(\gamma(s,Y_s)\frac{\partial{f}}{\partial{y}}(s,X_s^\alpha,Y_s,P_s)\biggr)^2\,ds}\biggr] <\infty
\end{align*}
because all the partial derivatives are bounded and using, respectively, the definition of the set $\mathcal{U}$,~\eqref{eqn:solutionP} and~\eqref{eqn:solutionY}.\\
Thus the first three integrals in~\eqref{eqn:mt_prop} are well defined and, according to the It\^o integral theory, they are martingales. Finally, the jump term in~\eqref{eqn:mt_prop} is a martingale too, being the function $f$ bounded.
\end{proof}

\section*{Acknowledgements}
The authors would like to thank Prof. Cristina Caroli Costantini for helpul technical suggestions.

\newpage
\bibliographystyle{apalike}
\bibliography{bib/biblio}

\begin{thebibliography}{}

\bibitem[Bai and Guo, 2008]{baiguo:optreins}
Bai, L. and Guo, J. (2008).
\newblock Optimal proportional reinsurance and investment with multiple risky
  assets and no-shorting constraint.
\newblock {\em Insurance: Mathematics and Economics}, 42:968--975.

\bibitem[Bass, 2004]{bass:2004}
Bass, R.~F. (2004).
\newblock Stochastic differential equations with jumps.
\newblock {\em Probab. Surveys}, 1:1--19.

\bibitem[Br\'emaud, 1981]{bremaud:pointproc}
Br\'emaud, P. (1981).
\newblock {\em Point Processes and Queues. Martingale dynamics}.
\newblock Springer-Verlag.

\bibitem[Cao and Wan, 2009]{caowan:optreins}
Cao, Y. and Wan, N. (2009).
\newblock Optimal proportional reinsurance and investment based on
  hamilton-jacobi-bellman equation.
\newblock {\em Insurance: Mathematics and Economics}, 45:157--162.

\bibitem[Gihman and Skorohod, 1972]{gihmanskorohod:sde}
Gihman, I. and Skorohod, A. (1972).
\newblock {\em Stochastic differential equations}.
\newblock Springer-Verlag.

\bibitem[Grandell, 1991]{grandell:risk}
Grandell, J. (1991).
\newblock {\em Aspects of risk theory}.
\newblock Springer-Verlag.

\bibitem[Gu et~al., 2010]{gu_yang:CEVoptreins}
Gu, M., Yang, Y., and Zhang, J. (2010).
\newblock Constant elasticity of variance model for proportional reinsurance
  and investment strategies.
\newblock {\em Insurance: Mathematics and Economics}, 46:580--587.

\bibitem[Heath and Schweizer, 2000]{heath:pde}
Heath, D. and Schweizer, M. (2000).
\newblock Martingales versus pdes in finance: An equivalence result with
  examples.
\newblock {\em Journal of Applied Probability}, (37):947--957.

\bibitem[Hipp, 2004]{hipp:stoch.control_applications}
Hipp, C. (2004).
\newblock Stochastic control with applications in insurance.
\newblock In {\em Stochastic Methods in Finance}, chapter~3, pages 127--164.
  Springer.

\bibitem[Irgens and Paulsen, 2004]{irgens_paulsen:optcontrol}
Irgens, C. and Paulsen, J. (2004).
\newblock Optimal control of risk exposure, reinsurance and investments for
  insurance portfolios.
\newblock {\em Insurance: Mathematics and Economics}, 35:21--51.

\bibitem[Li et~al., 2018]{lietal:robustXL}
Li, D., Zeng, Y., and Yang, H. (2018).
\newblock Robust optimal excess-of-loss reinsurance and investment strategy for
  an insurer in a model with jumps.
\newblock {\em Scandinavian Actuarial Journal}, 2018(2):145--171.

\bibitem[Liang and Bayraktar, 2014]{liangbayraktar:optreins}
Liang, Z. and Bayraktar, E. (2014).
\newblock Optimal reinsurance and investment with unobservable claim size and
  intensity.
\newblock {\em Insurance: Mathematics and Economics}, 55:156--166.

\bibitem[Liang et~al., 2011]{liangetal:optreins}
Liang, Z., Yuen, K., and Guo, J. (2011).
\newblock Optimal proportional reinsurance and investment in a stock market
  with ornstein–uhlenbeck process.
\newblock {\em Insurance: Mathematics and Economics}, 49:207--215.

\bibitem[Liang and Yuen, 2016]{liangetal:commonshock}
Liang, Z. and Yuen, K.~C. (2016).
\newblock Optimal dynamic reinsurance with dependent risks: variance premium
  principle.
\newblock {\em Scandinavian Actuarial Journal}, 2016(1):18--36.

\bibitem[Liu and Ma, 2009]{liuma:optreins}
Liu, B. and Ma, J. (2009).
\newblock Optimal reinsurance/investment problems for general insurance models.
\newblock {\em The Annals of Applied Probability}, 19:1495–1528.

\bibitem[Merton, 1969]{merton:inv}
Merton, R. (1969).
\newblock Lifetime portfolio selection under uncertainty: the continuous time
  case.
\newblock {\em Rev. Econ. Stat.}, (51).

\bibitem[Pascucci, 2011]{pascucci:PDEMartingale}
Pascucci, A. (2011).
\newblock {\em PDE and Martingale Methods in Option Pricing}.
\newblock Springer-Verlag Italia.

\bibitem[Schmidli, 2018]{schmidli:2018risk}
Schmidli, H. (2018).
\newblock {\em Risk Theory}.
\newblock Springer Actuarial. Springer International Publishing.

\bibitem[Sheng et~al., 2014]{shengrongzhao:optreins}
Sheng, D., Rong, X., and Zhao, H. (2014).
\newblock Optimal control of investment-reinsurance problem for an insurer with
  jump-diffusion risk process: Independence of brownian motions.
\newblock {\em Abstract and Applied Analysis}.

\bibitem[Xu et~al., 2017]{xu_zang_yao:markov_mod_fin}
Xu, L., Zhang, L., and Yao, D. (2017).
\newblock Optimal investment and reinsurance for an insurer under
  markov-modulated financial market.
\newblock {\em Insurance: Mathematics and Economics}, 74:7--19.

\bibitem[Young, 2006]{young:premium_princ}
Young, V.~R. (2006).
\newblock Premium principles.
\newblock {\em Encyclopedia of Actuarial Science}, (3).

\bibitem[Yuen et~al., 2015]{yuen_liang_zhou:commonshock}
Yuen, K.~C., Liang, Z., and Zhou, M. (2015).
\newblock Optimal proportional reinsurance with common shock dependence.
\newblock {\em Insurance: Mathematics and Economics}, 64:1--13.

\bibitem[Zhang et~al., 2009]{zhangyu:optreins}
Zhang, X., Zhang, K., and Yu, X. (2009).
\newblock Optimal proportional reinsurance and investment with transaction
  costs, i: Maximizing the terminal wealth.
\newblock {\em Insurance: Mathematics and Economics}, 44:473--478.

\bibitem[Zhu et~al., 2015]{zhuetal:reins_defaultable}
Zhu, H., Deng, C., Yue, S., and Deng, Y. (2015).
\newblock Optimal reinsurance and investment problem for an insurer with
  counterparty risk.
\newblock {\em Insurance: Mathematics and Economics}, 61:242--254.

\end{thebibliography}

\end{document}